%% file: mainArxiv.tex
\documentclass[thm-restate]{llncs}
\usepackage{bbding}
\input{macroslncs}

\begin{document}
\title{Weighing Timed Regular Languages:\\ The Final Step (long version)}
\author{Eugene Asarin\orcidID{0000-0001-7983-2202}\inst{1}\Envelope\and Aldric Degorre\orcidID{0000-0003-2712-4954}\inst{1}\and C\u at\u alin Dima\orcidID{0000-0001-5981-4533}\inst{2}
\and \\
Bernardo {Jacobo~Inclán}\orcidID{0009-0009-5323-7945}\inst{1,3} }

\institute{Université Paris Cité, CNRS, IRIF, Paris, France \email{\{asarin,adegorre\}@irif.fr} \and LACL, Université Paris-Est Créteil, France \email{dima@u-pec.fr}
\and LABRI, Université de Bordeaux, France \email{bernardo.jacobo-inclan@labri.fr}
}
\authorrunning{Asarin, Degorre, Dima, and Jacobo~Inclán}
\maketitle
\begin{abstract}
\input{abstract}
\end{abstract}

\section{Introduction}
   \input{introArxiv} 
\section{Background}\label{sec:back}
   \input{back}
   
\section{Measuring timed regular languages}\label{sec:band}
\input{band}
\section{Bandwidth computation for normal automata}   \label{sec:main}

\input{normal}

\input{sect-lower-bound}
\input{sect-upper-bound}
\section{Conclusions} \label{sec:concl}
\input{conclusion}
\vfill

\begin{credits}
This work was funded by ANR project MAVeriQ ANR-CE25-0012. The authors thank Thomas Colcombet for a valuable advice.
\end{credits}
\newpage
\bibliographystyle{splncs04}
\bibliography{entro}
\newpage
\appendix
\input{appendix}

\end{document}

%% file: macroslncs.tex
\usepackage{dsfont,mathtools,multicol,float,subcaption,eufrak,amssymb,listings,hyperref,underscore,complexity}
\usepackage{thm-restate}

\usepackage{tikz}
\usepackage{cleveref}
\usepackage[inline]{enumitem}

\usetikzlibrary {arrows.meta,graphs,quotes,positioning}
\usetikzlibrary{arrows,trees,backgrounds,automata,shapes,plotmarks,decorations,fit}
\tikzstyle{block}=[rectangle,draw, thin, inner sep=3pt, text centered,fill=orange!20!yellow!20] 
\tikzstyle{net}=[draw,cloud,fill=yellow!20,aspect=3,inner sep=1pt]
\tikzstyle{dev}=[draw,circle,fill=yellow!20,aspect=2,inner sep=1pt,minimum size=.6cm]
\tikzstyle{pre}=[<-,shorten <=1pt,>=stealth']
\tikzstyle{post}=[->,shorten >=1pt,>=stealth']
\tikzstyle{bi}=[<->,shorten >=1pt,shorten <=1pt, >=stealth']
\tikzstyle{every initial by arrow}=[initial text={},initial distance=1em,post]
\tikzstyle{every state}=[minimum size=0.6cm,fill=cyan!20!yellow!20]
\tikzstyle{transition}= [post,shorten >=1pt,node distance=2cm, inner sep=2pt,bend angle=20]

\newclass{\FPSPACE}{FPSPACE}

\spnewtheorem*{mainproblem}{Main Problem}{\bfseries}{\itshape}
\spnewtheorem{thm}{Theorem}{\bfseries}{\itshape}
\spnewtheorem{con}[theorem]{Construction}{\bfseries}{\itshape}
\crefname{thm}{Thm.}{Thms.}
\crefname{con}{Constr.}{Constrs.}
\crefname{proposition}{Prop.}{Props.}
\crefname{lemma}{Lem.}{Lemmas}
\crefname{definition}{Def.}{Defs.}
\crefname{corollary}{Cor.}{Cors.}
\crefname{figure}{Fig.}{Figs.}
\crefname{section}{Sect.}{Sections}
\crefname{equation}{Eq.}{Eqs.}

\newcommand{\conv}{\mathop{\mathrm{conv}}\nolimits}
\newcommand{\AF}{\mathrm{AF}}
\newcommand{\countb}{\mathrm{cb}}
\newcommand{\prevb}{\mathrm{pb}}
\newcommand{\lastreset}{\mathrm{r}}

\newcommand{\ent}{\mathcal{H}}
\newcommand{\capa}{\mathcal{C}}

\newcommand{\bandc}{\mathcal{BC}}
\newcommand{\bandh}{\mathcal{BH}}

\newcommand{\trans}[2][]{\xrightarrow[#1]{#2}}
\newcommand{\Trans}[1]{\cdot\!\cdot\!\!\xrightarrow{\!#1}}

\newcommand{\guard}{\mathfrak{g}}
\newcommand{\reset}{\mathfrak{r}}

\newcommand{\Net}{\mathsf{Net}}
\newcommand{\Sep}{\mathsf{Sep}}

\newcommand{\x}{\mathbf{x}}
\newcommand{\y}{\mathbf{y}}
\newcommand{\z}{\mathbf{z}}
\newcommand{\w}{\mathbf{w}}
\newcommand{\Kone}{\mathbf{K}}

\newcommand{\Word}{\mathop{word}}

\newcommand{\Runs}{\mathop{Runs}}

\newcommand{\untiming}{\mathop\mathbf{unt}}
\newcommand{\gap}{\mathop\mathrm{gap}}
\newcommand{\length}{\mathop{\mathrm{dur}}}

\newcommand{\real}{\mathds{R}}

\newcommand{\nat}{\mathds{N}}

\newcommand{\aut}{\mathcal{A}}

\newcommand{\dr}{\overrightarrow{d}}

\newcommand{\aaa}{node[blue]{$\bullet$}}
\newcommand{\bbb}{node[red]{$\bullet$}}

\newcommand{\RTA}{\textsc{RsTA}}
\newcommand{\CTA}{\textsc{SFA}}

%% file: abstract.tex
The bandwidth of a timed language characterizes the quantity of information per time unit (with a finite observation precision $\varepsilon$).   
The asymptotic behavior of the bandwidth as $\varepsilon \to 0$ classifies timed regular languages in three classes:
meager, normal, and obese.
Normal timed automata have a bounded frequency of events and some non-punctual transitions, and, up to now, were the only class of timed automata
for which no algorithm was available for computing their bandwidth. 
In this article, we compute the bandwidth of any such automaton in the form $\approx\alpha\log{1/\varepsilon}$. 
Our approach
reduces this problem to computing the best reward-to-cost ratio in a weighted finite graph
constructed from the given timed automaton.

\keywords{timed automata \and information theory \and bandwidth \and entropy}

%% file: introArxiv.tex
Studies of entropy (growth rate, information content) of regular and other formal languages raising to \cite{ChomskyMiller,shannon48}, constitute nowadays a textbook material \cite{perrin,marcus}, fundamental to symbolic dynamics, theory of codes, and practice of data coding and compression. In the long run, we are extending those studies to timed regular languages, with theoretical and practical ambitions.

In earlier works (see \cite{entroJourn}, and also \cite{timedCoding} for applications to information coding), we have defined and explored the notion of entropy of a timed language. It characterizes the quantity of information \emph{per event} contained in its words observed with some precision $\varepsilon$. 
In \cite{bw-our,3classes,CIAA,obese} we take a different, practically more relevant perspective. 
Given a timed regular language (represented by a timed automaton), observed with a precision $\varepsilon$, we want to know its \emph{bandwidth}, as introduced in \cite{bw-our} --- quantity of information in the words of the language  measured in bits \emph{per time unit}.

In \cite{3classes}, we have identified three classes of timed regular languages and proved that they constitute an exhaustive classification of timed regular languages. They are  illustrated by automata on \cref{fig:3kinds}: 
\begin{description}
    \item[meager] with bandwidth $O(1)$, which allows only encoding of information in discrete choices, and not in timing;
    \item[normal]\!\!\footnote{as opposed to meager and obese} with bandwidth $\Theta(\log (1/\varepsilon))$, where time delays (in quantity proportional to the duration) can be used to encode additional information;
    \item[obese] with bandwidth $\Theta(1/\varepsilon)$, events can happen with a very high frequency, and thus a huge amount of information can be encoded in timing.
\end{description}
\input{figures/example1}
As shown in \cite{3classes}, the three classes can be characterized by structural properties (presence or absence of certain patterns), but deciding to which class belongs a TA is \PSPACE-complete. 

The next research aim is a more quantitative analysis: for each class, find the main asymptotic term of the bandwidth:  in the form\footnote{all the logarithms in this article are base 2} $\alpha$, $\alpha\log (1/\varepsilon)$, and $\alpha/\varepsilon$ for meager, normal, and obese automata, respectively. 

In \cite{CIAA}, we solved the problem for (deterministic) meager automata and computed $\alpha$ using a new barycentric abstraction to transform a timed automaton into a (finite-state) simply-timed one, for which $\alpha$ can be characterized in terms of matrix algebra. In \cite{obese}, we solved it for obese automata by reduction to the maximum reward-per-time in the abstracted timed graph.

In this paper, we solve the problem for the last remaining case of a normal automaton. At first glance, the bandwidth of such a TA is $\alpha\log(1/\varepsilon)$  with $\alpha$ the maximal number of free (non-punctual) transitions per time unit. However, as we show in this work, the freedom of a transition is a subtle notion, which depends on the context (other transitions in the automaton) and the strategic choices (in which directions to move and in which not to move).  We formalize these ideas in terms of a refined bi-weighted graph associated with the automaton. Paths in this graph represent only certain runs in the timed automata, but are sufficient to generate its full bandwidth. Finally, $\alpha$ is characterized as the maximal reward-to-cost ratio in the refined graph.

Finding  (and proving) matching lower and upper bounds for the bandwidth was a challenging task. To establish the former, we considered only ``regular'' runs  of the automaton, which iterate many times one optimal cycle, and take every now and then another, resetting cycle to avoid a sort of saturation. The proof of the upper bound is much more involved, since it has to deal with all the possible behaviours of the automaton. It makes use of the algebraic version of Ramsey theory \cite{simon}, advanced reachability theory for timed automata \cite{puri}, and our previous findings on the bandwidth \cite{3classes}.  

The paper is structured as follows. In \cref{sec:back} we recall some useful notions. In \cref{sec:band} we define the bandwidth and  state the problem of its  computation. In \cref{sec:main} we  solve the problem for normal timed languages. We put together the results from \cite{3classes,CIAA,obese} and this article and discuss the perspectives in \cref{sec:concl}. In the main part of the article,  we present the algorithm, some intuitive explanations, and key lemmas; proof details are given in the  appendix.

%% file: figures/example1.tex
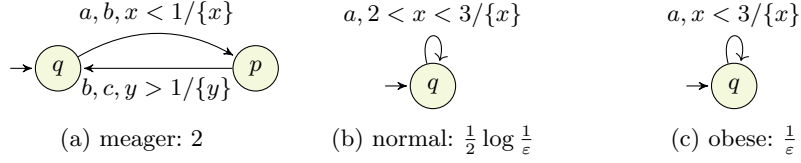
\begin{figure}[t]\small
\begin{subfigure}[t]{0.3\textwidth}
\centering
\begin{tikzpicture}
\node[state](q)[initial]   {$q$};
\node[state, right= 2cm of q](p) {$p$};
\draw [post] (q) edge [bend left, above] node {$a,b,x<1/\{x\}$} (p);
\draw [post]  (p) edge [ below] node {$b,c,y>1/\{y\}$} (q);
\end{tikzpicture}
\caption{meager: $2$}
\end{subfigure}
~
\begin{subfigure}[t]{0.3\textwidth}
\centering
\begin{tikzpicture}
\node[state](q)[initial] {$q$};
\draw  (q) edge [loop above] node {$a,2<x<3/\{x\}$} (q);
\end{tikzpicture}
\caption{normal: $\frac 1 2 \log\frac 1 \varepsilon$}
\end{subfigure}
~
\begin{subfigure}[t]{0.3\textwidth}
\centering
\begin{tikzpicture}
\node[state](q)[initial] {$q$};
\draw  (q) edge [loop above] node {$a,x<3/\{x\}$} (q);
\end{tikzpicture}
\caption{obese: $\frac 1 \varepsilon$}
\end{subfigure}
\caption{Timed automata and their approximate bandwidths \label{fig:3kinds}
    \vspace*{-10pt}}
\end{figure}

%% file: back.tex
\subsection{Timed words, languages and automata}
Timed automata have been introduced in \cite{AD} for modeling and verification of real-time systems.  We will use the main definitions in the following form:
\begin{definition}
    Given $\Sigma$, a finite alphabet of discrete events, a \emph{timed word} over $\Sigma$ is an element from $\left(\Sigma\times\real_+\right)^*$ of the form $w = (a_1,t_1)\dots (a_n,t_n)$, 
with $0 \leq t_1 \leq t_2 \cdots \leq t_n$. We denote $\length(w) = t_n$.
A \emph{timed language} over $\Sigma$ is a set of timed words over the same alphabet.
\end{definition}
For a set of variables $X$, let $G_X$ be the set of finite conjunctions of constraints of the form $x \sim b$ and $x \sim y+b$ with $x,y \in X$, $\sim \in \{<,\leq, >, \geq\}$ and $b \in \nat$.
\begin{definition}\label{def:TA}
 A  \emph{timed automaton} is a tuple 
$(Q, X, \Sigma, \Delta, S, I, F)$, where
\begin{itemize}
    \item $Q$ is the finite set of discrete locations;
    \item $X$ is the finite set of clocks;
    \item $\Sigma$ is a finite alphabet;
    \item $S,I,F: Q \rightarrow G_{X}$  define respectively the starting\footnote{this feature was not present in the original definition \cite{AD}}, initial, and final clock values for each location;
    \item 
    $\Delta \subseteq Q \times Q \times \Sigma \times G_{X} \times 2^X$ is the transition relation, whose elements are called edges. Given an edge $(q,q', a,\guard, \reset)$, we call $q$ its origin, $q'$ target, $a$ label, $\guard$ guard and $\reset$ reset. 
\end{itemize}

\end{definition}

A TA  is \emph{bounded} if there is a constant $M$ such that all the starting conditions $S(q)$ 
and the guards of all the edges imply that  clock values are smaller than $M$.

For $\x\in \real^n$ and  $d\in \real$ we write $\x+d$ as a shortcut for $(x_1+d,\dots, x_n+d)$.

The semantics of a timed automaton is given by a \emph{timed transition system}  whose states are 
tuples $(q,\x)$ composed of a location $q\in Q$ and a clock valuation (vector) $\x \in [0,\infty)^X$.
Each edge $\delta  = (q,q', a,\guard, \reset) \in \Delta$ generates many timed transitions, which are tuples $(q,\x) \trans[d]{\delta} (q', \x') $ where 
$ \x \models S(q)$, $\x+d \models \guard$ and $x'_c=0$ whenever $c\in\reset$ and $x'_c=x_c+d$ otherwise, provided that $\x' \models S(q')$. We also denote $\x[\reset]$ the operation of resetting 
the clocks in $\reset$, hence $\x' = (\x+d)[\reset]$. 

\emph{Paths} are sequences of edges that agree on intermediary locations. 
At the semantic level, they generate \emph{runs} which are sequences of timed transitions that agree on intermediary configurations.
An \emph{accepting run} is a run $\rho = (q_0,\x_0) \trans[d_1]{\delta_1} (q_1,\x_1)\cdots \trans[d_n]{\delta_n}(q_n,\x_n)$,
in which the first state satisifes $\x_0 \models I(q_0)$ and the last state satisfies $\x_n \models F(q_n)$.

For a given path $\pi$, its set of runs is $\Runs_\pi(\x,\y)$.
The language of $\pi$: $L_{\pi}(\x,\y)$ is the set of timed words associated with runs in  $\Runs_\pi(\x,\y)$, and the reachability relation $\mathop{Reach}_\pi$ is the set of pairs of $(\x,\y)\in\real_+^X\times\real_+^X$ such that $\Runs_\pi(\x,\y)\neq\emptyset$. Also, for $T\in\real_+$, $L_{\pi,T}(\x,\y)$ denotes the time bounded language $\left\{w\mid w\in  L_{\pi}(\x,\y)\wedge \length(w)\leq T\right\}$.
The sets ${\Runs_\pi}$, ${L_\pi}$ and ${L_{\pi, T}}$ denote, respectively, the unions of ${\Runs_\pi}(\x,\y)$, ${L_\pi}(\x,\y)$ and ${L_{\pi, T}}(\x,\y)$ for all possible values of $\x$ and $\y$.
The language of a timed automaton is the set of timed words associated with some accepting run and is denoted $L(\aut)$.
We will also need \emph{time-restricted languages} generated by some timed automaton: for any nonnegative real $T \in [0,\infty)$,
$L_T(\aut) = \{ w \in L(\aut) \mid \length(w)  \leq T\}$.

Furthermore, if $\delta_i = (q_i,q_{i+1},a_i,\guard_i,\reset_i)$, then 
the timed word associated with $\rho$ and initial time $t_0$ is defined as $\Word(\rho,t_0)\triangleq 
(a_1,t_0+d_1)(a_2,t_0+d_1+d_2)\ldots  (a_n,t_0+\sum_{i\leq n}d_i)$. 
We do not mention $t_0$ when it equals $0$.  
The \emph{length} of $\rho$ is $\length(\rho) = \length(\Word(\rho))$.

\paragraph{Regions and orbit graphs.}
Many algorithms related to timed automata utilize the finite abstraction of the timed transition system called the \emph{region construction} \cite{AD} that we briefly recall here. First, for $\xi \in \real$, we denote $\lfloor \xi \rfloor$ its integral part, and 
$\{\xi\}$  its fractional part. 
Then two clock vectors $\x$ and $\y$ are \emph{region-equivalent} iff
(1) $\lfloor x_c \rfloor = \lfloor y_c \rfloor$ for any clock $ c$;
(2) for any clock  $c$ we have that $\{x_{c}\}=0$ iff $\{y_{c}\} = 0$;
(3) for any two clocks  $c_1, c_2$ we have that $\{x_{c_1}\}\leq \{x_{c_2}\}$ iff $\{y_{c_1}\} \leq \{y_{c_2}\}$.
\emph{Regions} are equivalence classes of the region equivalence. 



We will use a special form of TA, similar to the  classical region automaton \cite{AD}, but typed again as a TA. We recall the definition \cite{entroJourn,3classes} in the form of \cite{obese}.
\begin{definition}\label{def:rs:bounded}
A \emph{region-split TA} (or \RTA) is a TA\ $(Q,X,\Sigma,\Delta,S,I,F)$, which is bounded and such that, for any location $q\in Q$: 
(1) $S(q)$ defines a non-empty region, called the \emph{starting region} of $q$;
(2) all states in $\{q\}\times S(q)$ are reachable from an initial\footnote{a state $(q,x)$ is initial whenever $x\models I(q)$ and final whenever $x\models F(q)$  } and co-reachable to a final one;
(3) for any edge $(q,q',a,\guard,\reset)\in\Delta$,  
$\left(\{S(q)+d\mid d\in \real_+\}\cap \guard\right)[\reset]=S(q')$, where we utilize the $(\cdot) [\reset]$ operator lifted to sets of clock valuations.
\end{definition}

\begin{lemma}\label{lem:2rta}
Any TA with an upper bound on some clock at every transition 
can be turned into a bounded region-split TA accepting the same language.  This transformation uses the same clocks 
and an exponentially larger set of locations.
\end{lemma}
The construction is sketched in \cite[full, B3]{obese}. 

\paragraph{Closed semantics.} We define the \emph{closed semantics} of an \RTA, in which an edge $\delta=(q,q',a,\guard,\reset)$ generates a transition
from $(q,\x)$ to $(q',\x')$ with duration $d>0$, when $\x\models \bar{S}(q)$,  $\x'\models\bar{S}(q')$,  $\x+d\models\bar\guard$ and $\x' = (\x+d)[\reset]$,
where $\bar S$ and $\bar \guard$ are obtained from $S$  and $\guard$ by replacing all ``$<$'' by ``$\leq$'' and all ``$>$'' by ``$\geq$''.

Definitions of $\overline{\Runs_\pi}(\x,\y)$, and $\overline{L_{\pi}}(\x,\y)$ and variants thereof are obtained similarly to the previous semantics, but based on closed semantics transitions, rather than usual transitions.
Additionally, we define the closed reachability relation
$\overline{\mathop{Reach}_\pi} = \{ (\x,\y)\in\real_+^X\times\real_+^X \mid \overline{\Runs_\pi}(\x,\y)\neq\emptyset\}$.



\subsection{Orbit graphs and Simon's theorem}
An important tool for characterizing reachability in timed automata is the orbit graph from \cite{puri}, that we define using monoid terminology.
\begin{restatable}[Orbit graph \cite{puri} and morphism $\gamma$ \cite{entroJourn}]{definition}{defOrbitMorphism}
\label{def:orbit:long}
Given a non-empty path $\pi$ from location $p$ to location $q$ in a \RTA, its \emph{orbit graph} $\gamma(\pi)$ is the triple $\langle p, G, q \rangle$, where $G$ is the graph whose vertices are the vertices of $S(p)$ and of $S(q)$ and such that there is an edge $(v,v')$ whenever $(v,v')\in \overline{\mathop{Reach}(\pi)}$.

For an empty path we put $\gamma(\epsilon)=1$, for a disconnected sequence of edges $\pi$ by definition $\gamma(\pi)=0$.
\end{restatable}
The image of $\gamma$ (containing orbit graphs, that can be composed with each other by interpreting them as relations over vertices, as well as special elements 0 and 1), denoted $O$, has a natural monoid structure; and $\gamma:\Delta^*\to O$ is a monoid morphism. 

We will use a powerful Simon's theorem to deal with such a morphism (see the original article \cite{simon} and a modern presentation in \cite{bojanczyk}). 

\begin{thm}[Simon]\label{thm:simon}
Let $\Gamma$ be a finite alphabet, $M$ a finite monoid, and $\beta:\Gamma^*\to M$ a monoid morphism. 
Then the sequence of sets 
\begin{equation}\label{eq:simon}
X_0=\Gamma\cup\{\epsilon\};\qquad 
X_{h+1}=X_h\cdot X_h\cup\bigcup_{\substack{e\in M\\ee=e}}(X_h\cap \beta^{-1}(e))^* 
\end{equation}
stabilizes at some finite level $H$ (in fact $H\leq 3|M|$) with $X_H=\Gamma^*$.
\end{thm}

Thus $\Gamma^*$ is stratified into $H$ levels, on level 0 there are only letters; and every word of level $i+1$ can be split into two arbitrary words of level $i$, or into many words of level $i$ all associated with a same idempotent element of $M$.

\subsection{$\varepsilon$-capacity and $\varepsilon$-entropy }
These two closely-related notions from  \cite{kolmoEpsilon} characterize the quantity of information needed to describe  (with precision $\varepsilon$) any element of a given space.

We first recall a couple of definitions:
\begin{definition}\label{def:sep:net}
Let $(X,d)$ be a metric space, subspace of some $Y$. A subset $M\subseteq X$ is called $\varepsilon$-\emph{separated} whenever $\forall x\neq y\in M : d(x,y)>\varepsilon$. 
A set $N\subseteq Y$ is an $\varepsilon$-\emph{net} for $X$ whenever $\forall x\in X \exists y\in N: d(x,y)\leq \varepsilon$.
\end{definition}
For a compact space $X$ and any $\varepsilon >0$, a finite $\varepsilon$-net must exist within $X$ itself, and cardinalities of $\varepsilon$-separated sets are bounded, which justifies the following definition.
\begin{definition}[Kolmogorov \& Tikhomirov]\label{def:ent:cap}
Given a compact metric space $X\subseteq Y$, let $\mathcal{M}_\varepsilon(X)$ be the maximal size of an $\varepsilon$-separated subset of $X$. Then the \emph{$\varepsilon$-capacity} of $X$ is  defined as 
$
\capa_\varepsilon(X)= \log \mathcal{M}_\varepsilon(X).
$
Let now $\mathcal{N}_\varepsilon(X)$ be the minimal size of an $\varepsilon$-net for $X$ (in $Y$). Then the \emph{$\varepsilon$-entropy} of $X$ (in $Y$) is
$
\ent^Y_\varepsilon(X)= \log \mathcal{N}_\varepsilon(X).
$
\end{definition}
As shown in \cite{kolmoEpsilon},  $\varepsilon$-entropy and $\varepsilon$-capacity are related by inequalities
\begin{equation}\label{prop:inequality}
\ent^Y_\varepsilon(X)\leq \capa_\varepsilon(X)\leq \ent^Y_{\varepsilon/2}(X).    
\end{equation}

\subsection{Optimal reward-to-cost ratio in finite graphs}\label{sec:alive}
Consider now a finite directed graph where to each edge $e$ is associated its \emph{reward} $r(e)$ and \emph{cost} $c(e)$. For a path $\pi$ its reward $r(\pi)$ is the sum of rewards of its edges, similarly for costs.  A classical optimization problem is maximizing the average reward-to-cost ratio $r(\pi)/c(\pi)$ over long or infinite paths. This ratio can be attained by iterating some optimal simple cycle. We summarize useful facts in a lemma, and  refer the reader to \cite{ratio} for a survey of exact and approximated algorithms for finding the optimal cycle and/or computing the best ratio.

\begin{lemma}\label{lem:rewcost}
   Given a finite directed graph with non-negative rewards and costs associated with edges, 
   let $\alpha$ be the maximal reward-to-cost ratio over all simple cycles. Then every cycle $\pi$ satisfies  
$
r(\pi)\leq\alpha c(\pi). 
$
Given a graph with $0/1$ rewards and costs,  $\alpha$ (as a rational number) is \FNL-computable\footnote{computable by a nondeterministic Turing machine in logarithmic space}.
\end{lemma}

%% file: band.tex
\subsection{Pseudo-distance on timed words}\label{sec:pseudo}
We are using the proximity measure on timed words introduced in \cite{distance}.
\begin{definition}\label{def:d}
The \emph{pseudo-distance} $d$ between timed words is defined as follows: given 
$w = (a_1,t_1)\dots(a_n,t_n) \text{ and } v = (b_1,s_1)\dots(b_m,s_m)$, 
\[
 \dr(w,v)\triangleq \max_{i\in 1..n}\min_{j\in 1..m} \{ |t_i-s_j|:a_i=b_j \}; 
 d(w,v) \triangleq  \max( \dr(w,v) , \dr(v,w) )
\]
 (with the standard convention $\min \emptyset = \infty$).
\end{definition}

This pseudo-distance allows a meaningful comparison of timed words with a different number of events. It is illustrated on \cref{fig:dist}.
Intuitively, two words are close to each other when they cannot be distinguished by an observer that reads the discrete letters of the word exactly 
(they can determine whether or not a letter has occurred), but with some imprecision w.r.t.~time. 
So, when two letters are very close to one another, the observer cannot determine which one came before the other, or not even how many times a letter was repeated within a short interval. 

\input{figures/distance.tex}

It is possible that $d$ fails to distinguish timed words, that is $d(w_1,w_2)=0$  but $w_1\neq w_2$; 
this is the case for $w_1=(a,1)(b,1)$ and $w_2=(b,1)(b,1)(a,1)$, that is why $d$ is  only a pseudo-distance.

\subsection{Bandwidth definition and problem statement}
\begin{definition}The $\varepsilon$-\emph{entropic bandwidth} and $\varepsilon$-\emph{capacitive bandwidth} of a timed language $L$ are defined respectively as $$\bandh_\varepsilon(L)=\limsup_{T\to\infty} \ent_\varepsilon(L_T)/T \text{ and }\bandc_\varepsilon(L)=\limsup_{T\to\infty} \capa_\varepsilon(L_T)/T.$$
\end{definition}
We remark that \cref{def:sep:net,def:ent:cap} were given for metric spaces, but they can be extended verbatim to timed words with our pseudo-distance $d$. Additionally, $\varepsilon$-nets (and thus  $\varepsilon$-entropy) are considered within the universal timed language $U$. 
We start by mentioning a few properties. First, it follows from \eqref{prop:inequality} that
\begin{equation}\label{eq:band:ineq}
\bandh_\varepsilon(L)\leq \bandc_\varepsilon(L)\leq \bandh_{\varepsilon/2}(L).    
\end{equation}
Second, the bandwidth is almost the same for $L_\pi$ and $\overline {L_\pi}$:
\begin{lemma}[\cite{3classes}]\label{lem:closed-semantics}
In an \RTA, for any  path $\pi$, clock vectors $x,y$, time bound $T>0$ and precision $\varepsilon>0$:
\begin{itemize}
 \item $\ent_\varepsilon(L_{\pi,T}(\x,\y)) = \ent_\varepsilon(\overline{L_{\pi,T}}(\x,\y))$;
 \item $\capa_\varepsilon(L_{\pi,T}(\x,\y)) \leq \capa_\varepsilon(\overline{L_{\pi,T}}(\x,\y))$;
 \item for any $0<\varepsilon'<\varepsilon$, $\capa_{\varepsilon}(\overline{L_{\pi,T}}(\x,\y))<\capa_{\varepsilon'}(L_{\pi,T}(\x,\y))$.
\end{itemize}
\end{lemma}

Third,  an upper bound on the bandwidth of timed languages holds, proved in different terms in \cite{distance}:
 \begin{proposition}[Maximal possible bandwidth]\label{prop:max}
 The bandwidth of the universal timed language $U_\Sigma$ satisfies $\bandc_\varepsilon(U_\Sigma)=\Theta(\frac{1}{\varepsilon})$.  Hence, for any timed language $L$, its bandwidth cannot be larger: $\bandc_\varepsilon(L)= O(\frac{1}{\varepsilon})$.
\end{proposition}

\begin{mainproblem}
Given a timed automaton $\aut$, compute the $\varepsilon$-bandwidths of its language asymptotically, as $\varepsilon\to 0$.
\end{mainproblem}
As stated in \cite{3classes} and discussed above, there are three classes of timed regular languages: meager with bandwidth $O(1)$, normal with bandwidth $\Theta(\log \frac1\varepsilon)$, and obese with the maximal possible bandwidth $\Theta(\frac1\varepsilon)$. In the rest of the paper, we compute the bandwidths of \textbf{normal} TA.  Precisely, given such a TA, we compute $\alpha$ such that both entropic and capacitive bandwidths are  $ \alpha\log\frac1\varepsilon
+o\left(\log\frac1\varepsilon\right)$. 

%% file: figures/distance.tex
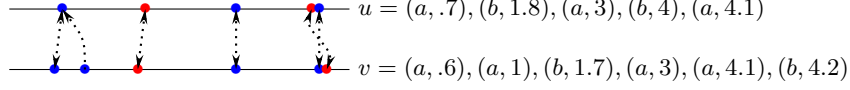
\begin{figure}[t]
\begin{center}
{
\usetikzlibrary {arrows.meta,positioning} 
\begin{tikzpicture}
\draw (0,.8) -- (.7,.8) \aaa  --(1.8,.8) \bbb  -- (3,.8) \aaa-- (4,.8) \bbb-- (4.1,.8) \aaa --
(4.5,.8)node[anchor=west]{$u=(a,.7),(b,1.8), (a,3), (b,4), (a,4.1)$};

\draw (0,0) -- (.6,0) \aaa   -- (1,0) \aaa--(1.7,0) \bbb  -- (3,0) \aaa-- (4.1,0) \aaa-- (4.2,0)\bbb --
(4.5,0)node[anchor=west]{$v=(a,.6),(a,1), (b,1.7), (a,3), (a,4.1),(b,4.2)$};

\draw [dotted,thick, -{Stealth[length=2mm,width=1mm]}]
(1,0).. controls (1,.4)  ..(.7,.8);

\draw [dotted,thick, {Stealth[length=2mm,width=1mm]}-{Stealth[length=2mm,width=1mm]}]
(.6,0) edge (.7,.8)
(1.7,0) edge (1.8,.8)
(3,0) edge (3,.8)
(4.1,0) edge (4.1,.8)
(4,.8) .. controls (3.9,.5) and (4.3,.3) .. (4.2,0)
;
\end{tikzpicture}
}
    \vspace*{-5pt}
\end{center}
\caption{
Pseudo-distance between two timed words (dotted arrows represent matching letters).  
$\protect\dr(u,v)=0.2;\  \protect\dr(v,u)=0.3$, thus  %
$d(u,v)=0.3$.
    \vspace*{-10pt}
}\label{fig:dist}
\end{figure}

%% file: normal.tex
\subsection{Informal discussion \& three examples} Intuitively, the coefficient $\alpha$ in the bandwidth of a normal automaton corresponds to the number of ``degrees of freedom'' per time unit. A degree of freedom corresponds to a transition for which the duration can be freely chosen in some interval. A natural idea would be to assign the reward $1$ to each ``free'' transition and compute the maximal reward-per-time ratio as described in \cite{alive}.
For example, the automaton $\aut_1$ on \cref{fig:two} has one free transitions every $2$ time units, which yields the bandwidth  of $0.5\log (1/\varepsilon)$. 

Unfortunately, it is not always  easy  to  classify transitions into free and non-free ones as illustrate two other examples on \cref{fig:two}.

\input{figures/twonormal.tex}
 $\aut_2$ seems to have 2 free transitions along the cycle of duration $1$, but in fact, $c$ is very constrained (it occurs closer and closer to the preceding $b$), and only $a$ is free, we will show later that the bandwidth is only $\log (1/\varepsilon)$. However, in  $\aut_3$, thanks to the slow loop labeled by $e$, taken quite seldom, the two events $a,c$ become free and yield the bandwidth $2\log (1/\varepsilon)$.

To distinguish between free and constrained transitions, we refine the states of the automaton. In a refined state,
we commit to moving only in some fixed directions (parallel to  some region faces, referred to as \emph{active} ones). 
A transition is considered free when it can take a variable time  while respecting the commitment.  Additionally, for each set of active faces in a location, we require existence of a \emph{resetting cycle} which repositions arbitrarily the clocks in any direction parallel to {active faces}.

Technically speaking, we  will characterize the bandwidth coefficient $\alpha$ of a normal\footnote{Our analysis applies to any timed automaton. For a meager one it returns $0$, for an obese one $\infty$.}  
timed automaton $\aut$ as the maximal reward-to-cost ratio of a finite graph  $G$ obtained by refinement of locations of $\aut$.    Refined states keep information on ``active faces"  
of timed regions used in the run. Reward corresponds to freedom degrees, cost roughly corresponds to time.  Paths in $G$ represent conservative runs of $\aut$ staying at a fixed distance from the active faces, but we will prove that this is sufficient to get the full bandwidth, as non-conservative runs tend to quickly drop some ``potential energy" in an unrecoverable way. 


\subsection{Pre-processing the timed automaton} 
We start with a  simple bandwidth-preserving transformation of the TA considered. We refer the reader to \cite{obese} for correctness proof and intuitive justification

\begin{con}[Adding heartbeat, \cite{obese}]
Given a TA, we add a new clock $h$  and a new letter $b$; next,  we add to each transition the constraint $h\leq 1$, and  to each location $p$ a self-loop $p\trans{b,h=1,\{h\}} p$. 
\end{con}
 As stated in \cite[Lemma 9]{obese}, this transformation preserves the bandwidth (intuitively, beat events every time unit do not bring any information). 
 At the  next step, we transform our TA into the bounded \RTA\ form (the heartbeat clock serves as a bounded clock required in \cref{lem:2rta}). Following  \cite{obese}, we call this normal form \CTA\ (\emph{standard-form automata})\footnote{in \cite{obese} we made yet another bandwidth-preserving transformation, not used here}.

 

\subsection{Refining the automaton}
We recall that geometrically clock regions are simplices in $\real^X$.
In a simplex, any non-empty set of vertices $C$ represents a simplex face $\conv C$.
\begin{definition}\label{def:facelist}
    A \emph{recurrent face list} for location $q$ is a non-empty list $(\ell_1, \dots,  \ell_k)$ with  $\ell_i$ disjoint non-empty sets of vertices of region $S(q)$, satisfying the following property. There exists a cycle  $\sigma$ from $q$ to $q$, such that its orbit graph $\gamma(\sigma)$ is complete (a clique) on the vertices in each $\ell_i$. 
\end{definition}
By Puri's theory \cite{puri},  
each clock valuation on the face $\conv(\ell_i)$ is reachable from each other on the same face via a run along the path $\sigma$, which should be the same for all faces.   

\begin{con}\label{con:main}
The \emph{refinement} of  an \CTA\ $\aut$ is a bi-weighted finite graph  $G=(Q',\Delta',r,c)$ with
\begin{itemize}
    \item $Q'$ the set of couples $(q,\ell)$ with $q\in Q$, $\ell$ a recurrent face list for $q$;
    \item $\Delta'$ containing an edge  $\delta'=(p,\ell) \to (q,\ell')$ whenever $|\ell|=|\ell'|$, and $\Delta$ contains $\delta= p\trans{a,\guard/\reset}q$, with orbit graph $\gamma(\delta)$ containing for each $i\leq |\ell|$
\begin{itemize}
    \item an edge from each vertex of $\ell_i$ to some vertex of $\ell'_i$;
    \item an edge to each vertex of $\ell'_i$ from some vertex of $\ell_i$.
\end{itemize}
\item
The reward $r(\delta')$ is $1$ if for some $i\leq |\ell|$ exists $v\in \ell_i$ and some $l<u$ such that for any  delay $d\in (l,u)$ it holds that
$v+d\models \guard  \text{ and } (v+d)[\reset]\subseteq \conv(\ell'_i)$,
otherwise $r(\delta')=0$.
\item The cost $c(\delta')$ is $1$ whenever $\reset\ni h$ and $0$ otherwise.
\end{itemize}
\end{con}
Informally, the reward for a transition is $1$ whenever the transition is non-punctual at least on one recurrent face. The cost is $1$ whenever it is a heartbeat transition.


By \cref{lem:rewcost}, one can compute the optimal cycle in $G$ and its reward-to-cost ratio, a rational number. Now we can formulate the main result of this article.

\begin{thm}\label{thm1}
Given an \CTA\ $\aut$, let  $G$ be its refined graph, with $\alpha=\alpha(G)$ its maximal reward-to-cost ratio.
\begin{itemize}
    \item If $G$ does not contain cycles or if $\alpha=0$ then $\aut$ is meager;
    \item if $\alpha =\infty$, then $\aut$ is obese;
    \item otherwise it is normal with   
    $$\bandh_\varepsilon(L(\aut))= \alpha\log\frac1\varepsilon +o\left(\log\frac1\varepsilon\right) \text{ and } \bandc_\varepsilon(L(\aut))= \alpha\log\frac1\varepsilon +o\left(\log\frac1\varepsilon\right).$$
\end{itemize}
\end{thm}
A proof sketch is given in subsections 4.5 and 4.6, and full proofs can be found in the ArXiV version of the article. 
As for complexity, it can be estimated using \cref{lem:rewcost} and knowing that the \CTA-form and graph $G$ are exponential in the size of the original automaton.   
\begin{corollary}
Given a timed automaton, the coefficient $\alpha$ as in the theorem is a rational number computable in \FPSPACE. 
\end{corollary}


\subsection{Examples of bandwidth computation}
We start by sketching the application of the algorithm above to three normal automata from \cref{fig:two}.

\input{figures/processing.tex}

For $\aut_1$, we first add a heartbeat (for simplicity we group together $b$ and $c$), the result is given on \cref{fig:processing},a. We have also explicitly drawn the starting regions of each location.

In this automaton, trivially, all starting regions are recurrent, which yields the cycle in the refined graph presented on \cref{fig:processing},b.  The reward is $1$ on the non-punctual $a$-transition, the cost is $1$ on both $b$-transitions, which yields the ratio $1/2$. Other cycles in the graph (not presented on the figure) correspond to the choice of one-point active face in $p$, which leads to $0$ reward everywhere. We conclude that $\alpha=1/2$, i.e.~$\bandc_\varepsilon(\aut_1)\approx 1/2\cdot \log(1/\varepsilon)$.  

In $\aut_2$, heartbeat is already present and we can build the refined graph directly. Unfortunately, the  starting regions are not recurrent and graph $G$ contains only active faces of smaller dimensions. An optimal cycle with $\alpha=1/1$ is shown in \cref{fig:processing},c. 

In $\aut_3$, the ``slow" loop $e$ makes the starting region of $q$ recurrent, hence the cycle $acb$ 
(which does not involve $e$!) provides the reward-to-cost ratio $\alpha=2/1$.

\subsection{Lower bound: coding with two cycles}
The proof idea for the lower bound is as follows. We consider the optimal cycle $\pi$ from $(q,\ell)$ to itself in the weighted graph $G$ with reward-to-cost ratio $\alpha$. We call region faces spawned by $\ell_i$ active faces. By construction, they are recurrent:  there exists a ``resetting'' cycle $\sigma$ as in \cref{def:facelist} which allows arbitrary movement of the clock vector within each face. We show that runs following $(\pi^\theta \sigma)^*$ for some big value of $\theta $ produce the bandwidth arbitrarily close to $\alpha$. To that aim, we iterate the cycle $\pi$ keeping the clock vector $\x$ on a cross-section of the region parallel to all active faces. We use all the free transitions (i.e., those with the reward $=1$) to encode $\approx\log(1/\varepsilon)$ bits of information by choosing an $\varepsilon$-discrete duration. At some moment, this could lead to a saturation, making further encoding impossible, like in $\aut_3$. At this moment, we use the cycle $\sigma$ to reset all clocks and restart. In the long run, we encode $\approx\alpha\log(1/\varepsilon)$ bits/time unit while evolving in $\pi$, and less when resetting clocks on $\sigma$. However, if $\varepsilon$ is small, $\theta $ can be chosen big, and the resulting bandwidth will be arbitrarily close to  $\alpha\log(1/\varepsilon)$.  
Technically speaking, the strategy described above is used to build a  large $\varepsilon$-separated set for words that are generated following  $(\pi^\theta \sigma)^*$. 




%% file: figures/twonormal.tex
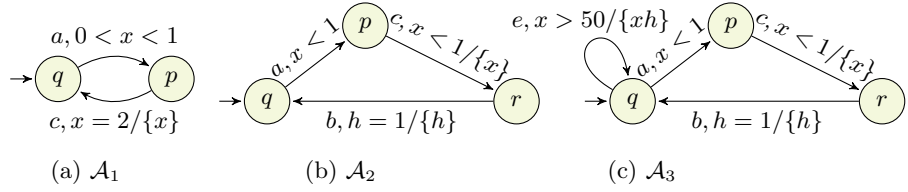
\begin{figure}[t]
\begin{subfigure}[t]{0.2\textwidth}
    \centering
    \begin{tikzpicture}\small
    \node[state](q0)[initial] at (-1.5,0.5) {$q$};
    \node[state](p0)  at (0,0.5)  {$p$};
    \draw [post] (q0) edge [bend left, above] node {$a,0<x<1$} (p0);
    \draw [post] (p0) edge [bend left, below] node {$c,x=2/\{x\}$} (q0);
    \end{tikzpicture}
    \caption{$\aut_1$}
\end{subfigure}
 ~   
\begin{subfigure}[t]{0.3\textwidth}
    \centering
    \begin{tikzpicture}\small    
\node[state](q)[initial] at (1,0)  {$q$};
\node[state](p)  at (2.3,1)  {$p$};
\node[state](r)  at (4.3,0)  {$r$};

\draw [post] (q) edge [sloped,above] node {$a,x<1$} (p);
\draw [post] (p) edge [sloped, above] node {$c,x<1/\{x\}$} (r);
\draw [post]  (r) edge [sloped, below] node {$b,h=1/\{h\}$} (q);
    \end{tikzpicture}
    \caption{$\aut_2$}
\end{subfigure}
 ~   
 \begin{subfigure}[t]{0.3\textwidth}
    \centering
    \begin{tikzpicture}\small    
\node[state](q)[initial] at (1,0)  {$q$};
\node[state](p)  at (2.3,1)  {$p$};
\node[state](r)  at (4.3,0)  {$r$};

\draw [post] (q) edge [sloped,above] node {$a,x<1$} (p);
\draw [post] (p) edge [sloped, above] node {$c,x<1/\{x\}$} (r);
\draw [post]  (r) edge [sloped, below] node {$b,h=1/\{h\}$} (q);
\draw [post] (q) edge [loop above,min distance=1cm, in=100, out =150] node {$e,x>50/\{xh\}$} (q);
    \end{tikzpicture}
    \caption{$\aut_3$}
\end{subfigure}

    \caption{Three normal automata}
    \label{fig:two}
    \vspace*{-10pt}
\end{figure}

%% file: figures/processing.tex
    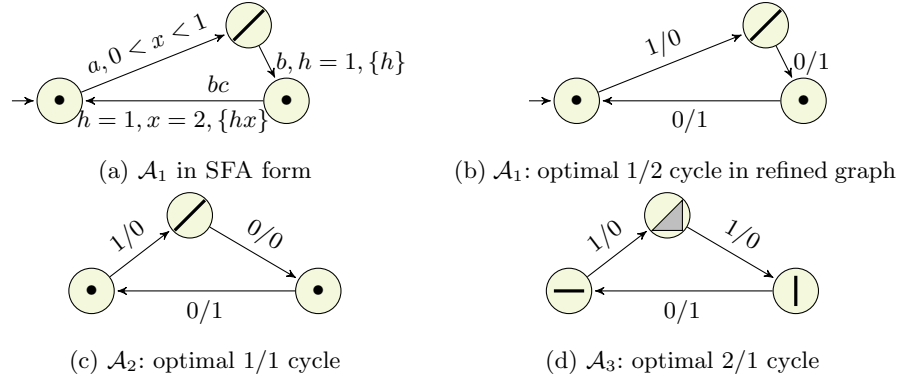
\begin{figure}[t]
\begin{subfigure}[t]{0.5\textwidth}
    \centering
     \begin{tikzpicture}
    \small
            \node[state](q)[initial] at (0,0) {$\bullet$};
    \node[state](p)  at (2.5,1)  {$ $};
    \node[state](s)  at (3,0)  {$\bullet$};
    \draw[very thick] (2.3,0.8) -- (2.7,1.2);
    \draw [post] (q) edge [above,sloped] node {$a,0<x<1$} (p);
    \draw [post] (p) edge [right] node {$b,h=1,\{h\}$} (s);
    \draw [post] (s) edge [below]  node{$h=1,x=2,\{hx\}$}  node[above,near start] {$bc$} (q);
        \end{tikzpicture}
        \caption{$\aut_1$ in \CTA\ form}
    \end{subfigure}
    ~
    \begin{subfigure}[t]{0.48\textwidth}
    \centering  
            \begin{tikzpicture}
    \small
            \node[state](q)[initial] at (0,0) {$\bullet$};
    \node[state](p)  at (2.5,1)  {$ $};
    \node[state](s)  at (3,0)  {$\bullet$};
    \draw[very thick] (2.3,0.8) -- (2.7,1.2);
    \draw [post] (q) edge [above,sloped] node {1/0} (p);
    \draw [post] (p) edge [right] node {0/1} (s);
    \draw [post] (s) edge [below] node {0/1} (q);
        \end{tikzpicture}
        \caption{$\aut_1$: optimal $1/2$ cycle in  refined graph}
    \end{subfigure}
    \\\\
    \begin{subfigure}[t]{0.5\textwidth}
    \centering
    \begin{tikzpicture}\small    
\node[state](q) at (1,0)  {$\bullet$};
\node[state](p)  at (2.3,1)  {$ $};
\draw[very thick] (2.1,0.8) -- (2.5,1.2);
\node[state](r)  at (4,0)  {$\bullet$};

\draw [post] (q) edge [sloped,above] node {1/0} (p);
\draw [post] (p) edge [sloped, above] node {0/0} (r);
\draw [post]  (r) edge [sloped, below] node {0/1} (q);
    \end{tikzpicture}
    \caption{$\aut_2$: optimal 1/1 cycle}
\end{subfigure}
 ~   
     \begin{subfigure}[t]{0.5\textwidth}
    \centering
    \begin{tikzpicture}\small 
    \node[state](q) at (1,0)  {$ $};
\draw[very thick] (0.8,0) -- (1.2,0);

\node[state](p)  at (2.3,1)  {$ $};
\draw[fill=lightgray] (2.1,0.8) -- (2.5,1.2) -- (2.5,0.8)-- cycle;
\node[state](r)  at (4,0)  {$ $};
\draw[very thick] (4,-0.2) -- (4,0.2);

\draw [post] (q) edge [sloped,above] node {1/0} (p);
\draw [post] (p) edge [sloped, above] node {1/0} (r);
\draw [post]  (r) edge [sloped, below] node {0/1} (q);

    \end{tikzpicture}
    \caption{$\aut_3$: optimal 2/1 cycle}
\end{subfigure}
        \caption{Computing the bandwidth of $\aut_1, \aut_2, \aut_3$ from \cref{fig:two}. Starting regions on (a) and active faces  on (b-d) are sketched within each location.}
        \label{fig:processing}
    \vspace*{-10pt}
    \end{figure}

%% file: sect-lower-bound.tex
Hereinafter, we consider an \CTA\ $\aut$ with a refinement graph $G$, and fix notations for their size parameters. So, let $N$ be the number of locations, $M$  a bound on its clock values, $K-1$ the number of clocks; $O$ the monoid of orbit graphs, and finally $H\leq 3 \#O$ the stabilization level in \cref{thm:simon} for this monoid. Later on we will introduce several other parameters $K_0,K_1,\dots $ depending only on $\aut$ and $G$ (and not on $\varepsilon$ and $T$).

\begin{restatable}{definition}{defRoute}
A \emph{route} $\pi,\pi'$ is a couple of paths $\pi=\delta_1\dots\delta_n$ in $\aut$ and $\pi'=\delta'_1\dots\delta'_n$ in $G$, where each $\delta'_i$ refines $\delta_i$ as described in \cref{con:main}. 
\end{restatable}

To prove the lower bound, we first relate the $\varepsilon$-capacity of a cycle $\pi$ in $\aut$ to the reward $r(\pi)$ and the cost $c(\pi)$ of its refinement in $G$. Roughly, along such a cycle in time $\approx c(\pi)$, one can do $r(\pi)$ choices of timing with $\approx 1/\varepsilon$ possibilities for each choice.  

The precise setting is as follows.
\begin{equation}\label{eq:set}
\text{\parbox{0.93\textwidth}{
For a given location $q$ and face list $\ell$ containing $k$ elements:
\vspace*{-7pt}
\begin{itemize}
    \item $\pi,\pi'$ is a cyclic route from $(q,\ell)$ to itself;
    \item  $\sigma$  is a cyclic path from $q$ to $q$ in $\aut$ such that all $\ell_i$ are cliques in $\gamma(\sigma)$. 
\end{itemize}
}}    \vspace*{-5pt}
\end{equation}

The following lemma formalizes the idea of iterating along the cyclic route $\pi,\pi'$ by producing an 
$\varepsilon$-separated set with the appropriate size:

\begin{restatable}{lemma}{lemSepOne}
\label{lem:sep:one}
In setting \eqref{eq:set},  there exist $\beta>0$ (depending on $\pi,\pi'$),  a polytope $R\subseteq S(q)$, and a point $\x^* \in R$,  such that for all $\varepsilon>0$ small enough there exists an $\varepsilon$-separated set  $\Sep\subset L^\aut_{\pi, t'} (\x^*,R)$ of cardinality $\geq \beta /\varepsilon^{r(\pi) }$. Here $t'=c(\pi)+1$. 
\end{restatable}

The next lemma then formalizes the idea of ``resetting all clocks" on the cyclic path $\sigma$, 
which will allow the previous $\varepsilon$-separated set to be repeated again after resetting:  
\begin{restatable}{lemma}{lemReturn}
\label{lem:return}
In setting \eqref{eq:set}, for any $\y\in R$ there exists a run $\rho^y$  along $\sigma$ of duration $\leq c(\sigma)+1$ from $(q,\y)$ to $(q,\x^*)$.
\end{restatable}

Combining the two previous lemmas, we obtain 
\begin{restatable}{corollary}{corSigmaPi}\label{cor:one}
In setting \eqref{eq:set}, there exist $\beta>0$, and a point $\x^* \in R$,  such that for all $\varepsilon>0$ small enough there exists an $\varepsilon$-separated set $\Sep'$ in $L^\aut_{\pi\sigma, t''} (\x^*,\x^*)$ of cardinality $\geq \beta /\varepsilon^{r(\pi)} $, where $t''=c(\pi)+c(\sigma)+2$.
\end{restatable}
We apply this to the cycle $\pi^\theta$ (with an appropriately selected large $\theta$), and
iterating $\approx T/t''$ times the set $\Sep'$ we obtain a large $\varepsilon$-separated set in $L_T(\aut)$, which allows us to prove the following:
\begin{restatable}[Lower bound]{proposition}{propLowerBound}
\label{lem:lower}  If $\alpha <\infty$, then for every $\varkappa>0$, and for $\varepsilon>0$ small enough, 
$\bandc_\varepsilon(L(\aut))\geq (\alpha-\varkappa)\log\frac1\varepsilon$.
\end{restatable}
The following proposition complements the results from \cite{3classes}:
\begin{restatable}{proposition}{propAlphaInfty}
\label{lem:obese}
If $\alpha =\infty$, then $\aut$ is obese.
\end{restatable}

%% file: sect-upper-bound.tex
\subsection{Upper bound: factorizing and approximating }
For the rest of the paper, given a path $\sigma$, we denote $\overline\sigma$ its \emph{closed semantics},
obtained by transforming each strict inequality into a non-strict one. 
Results from \cite{3classes} show that this transformation of a path does not modify its bandwidth. 

The proof of the upper bound proceeds in several steps. 
The first step uses \cref{thm:simon} (of Simon)
to factorize every run into conservative cycles (evolving in parallel to some active faces), and arbitrary intermediate lossy chunks (the number and total duration of the latter depend only on $\varepsilon$). The idea is that there is a Lyapunov function, a potential, which is preserved by conservative runs, but irremediably and quickly decreased by the lossy chunks.


By \cref{thm:simon}, each path in $\aut$ belongs to the level  $X_H$ of the sequence
\eqref{eq:simon}. 
Similarly, runs of $\aut$ belong to level $R_H$ of the sequence
$R_0=R_\Delta\cup\{\epsilon\}$, 
$R_{h+1}=R_h\cdot R_h\cup\bigcup_{\substack{e\in O\\ee=e}}(R_h\cap \Runs(\gamma^{-1}(e)))^*.$
Note that $H$ is polynomial in the size of the \CTA\ $\aut$. 

For a given idempotent orbit graph $e\in O$, let $C_i$ be the strongly-connected components (cliques or transient singletons) within its vertices; we enumerate them in a topological order w.r.t. the graph $e$. Next, we introduce a variant of barycentric coordinates \cite{bary}: 
any point  $\x\in S(q)$ can be represented as $\x=\sum_i\lambda_i \x_i$  with  $\x_i\in \conv (C_i)$, and $ \lambda_i\geq 0$ such that $\sum_i\lambda_i=1$;  the coordinates   $\vec\lambda(\x)=(\lambda_i)$ characterize the position of $\x$ with respect to active faces $\conv{C_i}$.
According to \cite{puri,entroJourn}, the coordinates   $\vec\lambda (\x)$ can only decrease (w.r.t. the lexicographic order $\succeq$) when traversing a run in $\Runs(\gamma^{-1}(e))$.

Then, the next ingredient of our reasoning are runs where the inequality $\vec\lambda (\x^1)\succeq \vec\lambda (\x^2)$ turns into equality (or is $\varepsilon$-close to equality). The next definition formalizes these:

\begin{restatable}{definition}{defConservativeAndEps}
Let  $\sigma\in \gamma^{-1}(e)$  be a cycle from $q$ to itself. A run 
from $(q,\x^1)$ to $(q,\x^2)$ along $\overline\sigma$ is 
a \emph{conservative} $e$-run 
whenever $\vec{\lambda}(\x^1)= \vec{\lambda}(\x^2)$. 
It is  $\varepsilon$-\emph{conservative} whenever $||\vec{\lambda}(\x^1) -\vec {\lambda}(\x^2)||<\varepsilon $.
\end{restatable}

Denote also $f_\aut$ the real-valued function with $f_\aut(\varepsilon) = (1+1/\varepsilon)^{KH}$.
The following result plays a central role in the construction of the desired $\varepsilon$-net:
\begin{restatable}{proposition}{simonDecompConserv}
\label{lem:fac}
For each $\varepsilon>0$, 
any timed word $w\in L(\aut)$  can be  factorized as $u_0v_1u_1\dots v_n u_n$ with 
$n\leq f_\aut(\varepsilon)$ and $\sum_{i=0}^n\length(u_i) < f_\aut(\varepsilon)$ 
and all  $v_i$ are traces of $\varepsilon$-conservative runs.
\end{restatable}

\begin{proof}[idea]
Denote  $b=(1+1/\varepsilon)^K$. 
We may prove, by induction on $h$, that each timed word which is a trace of $R_h$ can be factorized as 
$u_0v_1u_1\dots v_n u_n$
with $n\leq b^h $, $\sum\length(u_i)\leq b^h$, 
and all $v_i$ are traces of $\varepsilon$-conservative runs.


We give here the proof idea for the inductive case, the subcase of 
$w$ trace of a run $\rho=\rho_0\dots \rho_m$ with all paths $\rho_j\in (R_h\cap \Runs(\gamma^{-1}(e)))$, for some integer $m$.
In this situation, the hypothesis of $e$ being an idempotent, 
combined with results on Lyapunov functions on the orbit graph \cite{entroJourn}.
imply that, if $(q,\x_i)\Trans{\rho_i}(q,\x_{i+1})$, then 
$\lfloor\vec{\lambda}(\x_i)/\varepsilon\rfloor\succeq\lfloor\vec{\lambda}(\x_{i+1})/\varepsilon\rfloor$.

On the other hand, for each fixed $\varepsilon$, 
the quantities $\lfloor \vec{\lambda}(\x_i)/\varepsilon \rfloor$ may have at most $b=(1+1/\varepsilon)^{K}$ different values. 
But then $\rho$ can be re-decomposed as
\[
(q,\y_0)\Trans{\rho'_0}(q,\y'_0)\Trans{\rho''_0}(q,\y_1)\Trans{\rho'_1}(q,\y'_1)\cdots(q,\y_j)\Trans{\rho'_j}(q,\y'_j),
\]
such that $\y_0=\x_0$ and $\y'_j=\x_j$, $j\leq b$, 
$\lfloor\vec{\lambda}(\y_i)/\varepsilon\rfloor=\lfloor\vec{\lambda}(\y'_i)/\varepsilon\rfloor$ 
i.e., $\rho'_i$ is $\varepsilon$-conservative, and 
each of $\rho''_i$ belongs to  $R_h$. 
The induction hypothesis applied to $\rho''_i$ then allows us to conclude. \qed
\end{proof}
Thanks to the following proposition, instead of $\varepsilon$-conservative runs we can consider only conservative, with better properties:
\begin{restatable}{proposition}{epsilonCon}\label{lem:eps:con}There exists $K_0>0$ depending only on $\aut$, such that any  $\varepsilon$-\emph{conservative} $e$-run is $K_0\varepsilon$-close to a conservative $e$-run.
\end{restatable}
Next, we show that conservative runs necessarily follow cycles in the graph $G$ (``follow" is formalized in the next definition). 
Such runs satisfy the inequality $r(\pi)\leq \alpha c(\pi)$, i.e.~the number of free transitions in time $\Theta$ is bounded by $\alpha\Theta$. This fact is used to construct a small $\varepsilon$-net for traces of all conservative runs. 

\begin{restatable}{definition}{runFollowsRoute}
\label{def:follow}
Fix a route $\pi,\pi'$  with 
$\pi=\delta_1\cdots\delta_n$
and 
$\pi'=(q_0,\ell^0)\to(q_1,\ell^1)\to\cdots\to(q_n,\ell^n)$. 
A run $\rho= (q_0,\x^0) \trans[d_1]{\delta_1} (q_1,\x^1)\cdots \trans[d_1]{\delta_1}(q_n,\x^n)$ in $\overline{\Runs_\pi}$ \emph{follows} this route
whenever: 
\begin{itemize}
    \item all the lists $\ell^j$ have the same length $k$;
    \item there exists a decomposition $\vec\rho=\sum_{i=1}^k \lambda_i\vec\rho_i$ with all $\lambda_i>0$, $\sum\lambda_i=1$, all $\rho_i$ runs of $\aut$ in $\overline{\Runs_\pi}$;
    \item each $\rho_i$ hits the $i$-th  active face on all the intermediate states, that is, it has a form
    $\rho_i= (q_0,\x^0_i) \to (q_1,\x^1_i)\cdots \to (q_n,\x^n_i)$ with
    $\x^j_i\in \conv\ell^j_i$ for $j=0..n$. 
\end{itemize}
\vspace*{-10pt}
\end{restatable}

\begin{restatable}{proposition}{conservRunFollowsRoute}
\label{lem:realize}Every conservative run follows some route. 
\end{restatable}

Routes of cost and reward $0$ need to be handled properly when creating an $\varepsilon$-net of the appropriate size. 
They satisfy the following pumping lemma (with an appropriate definition of $0$-loop): 
\begin{restatable}{lemma}{constBforZeroLoop}
\label{lem:pumping}
For any \CTA\ $\aut$ there exists a constant B, such that any route $\pi,\pi'$ with $|\pi|>B$ and $c(\pi)=r(\pi)=0$ contains a $0$-loop.
\vspace*{-5pt}
\end{restatable}

%
\begin{restatable}{definition}{defProperRoute}
A route is \emph{proper} if it does not contain $0$-loops.
\vspace*{-5pt}
\end{restatable}


\begin{restatable}{lemma}{countNumberCyclicRoutes}
\label{lem:count:routes}
Given a normal \CTA, there exist constants $K_i$ such that the number of proper cyclic routes  of  a cost $\leq T$ is at most $K_1 2^{K_2T}$. The number of transitions in any such route is at most $K_{10}T+K_{11}$. 
\end{restatable}


It is possible to transform any route into a proper one by removing all 0-loops. In the following lemma, 
$\phi$ is an auxiliary multi-valued function on timed words (it inserts some instantaneous events in $w$).

\begin{restatable}{lemma}{fluffFollowingRoute}
\label{lem:fluff}
For every route  $\pi,\pi'$ exists a proper route $\hat \pi,\hat\pi'$ such that for any trace $w$ of a run following the former exists  $\hat{w}$ trace of a run following the latter 
such that $d(w,\phi(\hat{w}))=0$.
\end{restatable}

We are now in a position to build the desired $\varepsilon$-net. We do this by defining
$\varepsilon$-nets for the conservative factors of runs, which are then 
interleaved with $\varepsilon$-nets for intermediate chunks.

\begin{restatable}[$\varepsilon$-net for  words following one route]{con}{epsNetWordsOneRoute}
\label{con:con:1} Given  a 
proper route $\pi,\pi'$  
from $(q,\ell)$ to itself, and an initial time interval $[\tau,\tau+1)$, we build a set 
$\Net^{\mathbf{fol}}_\varepsilon(\pi,\pi',\tau)$
 containing all words $v=(a_1,s_1)\dots(a_n,s_n)$, with $\varepsilon$-discrete timestamps $s_i\in \varepsilon\nat$, that satisfy the following properties:
\begin{enumerate}
    \item if $s_{i}-\tau>2$ and $a_i= b$, then $t_i\in s_{\prevb_i} + [1-\varepsilon,1+\varepsilon] $; 
    \item $s_i\in [\tau+\countb_i-1-\varepsilon,\tau+\countb_i+2+\varepsilon]$;
    \item if $s_{i-1}-\tau>M$ and $r(\delta')=0$, then  
    $s_i\in s_{i-1}+g_i(\x_{i-1}) +[-\nu\varepsilon-\varepsilon,\nu\varepsilon+\varepsilon]$, where
    $\x_i^k=s_i-s_{\lastreset_i^k}$. 
\end{enumerate}
\end{restatable}
Here $\countb_i$, $\prevb_i$ and $\lastreset_i^k$ depend on the path $\pi$ only, and $g_i$ is a uniformly $\nu$-Lipshitz continuous function.

\begin{restatable}{lemma}{lemmaEpsNetWordsOneRoute}
\label{lem:net:con:1} 
The set $\Net^{\mathbf{fol}}_\varepsilon(\pi,\pi',\tau)$ defined in \cref{con:con:1} is an $\varepsilon$-net for all  words traces of runs following the route $\pi,\pi'$ (starting within $[\tau,\tau+1)$). 
\end{restatable}



\begin{restatable}[$\varepsilon$-net for words following proper routes]{con}{epsFollowProperRoutes}
\label{con:con:T}
We define\\
$
\Net^{\mathbf{fol}}_{\varepsilon,\Theta}(\tau)=\bigcup_{\pi,\pi'}\Net^{\mathbf{fol}}_\varepsilon(\pi,\pi',\tau),
$
where $\pi,\pi'$ range over all proper routes  with cost $\leq\Theta$. 
\end{restatable}

\begin{restatable}{lemma}{indeedEpsFollowProperRoutes}
\label{lem:net:con}
The set $\Net^{\mathbf{fol}}_{\varepsilon,\Theta}(\tau)$ defined in \cref{con:con:T} is an $\varepsilon$-net for all words starting at $t_0\in[\tau,\tau+1)$ and following some proper  route of cost $\leq\Theta$. 
\end{restatable}

\begin{restatable}[$\varepsilon$-net for conservative words]{con}{epsNetConservWords}
\label{con:con:con} Let now
\[
\Net^{\mathbf{cons}}_{\varepsilon,\Theta}(\tau)=\phi\left(\Net^{\mathbf{fol}}_{\varepsilon,\Theta+1}(\tau)\right).
\]
\end{restatable}

\begin{restatable}{lemma}{coroEpsNetCoservWords}
\label{cor:net:fluff}
The set $\Net^{\mathbf{cons}}_{\varepsilon,\Theta}(\tau)$ is an  $\varepsilon$-net for all conservative words of $L_\Theta(\aut)$ starting at $t_0\in[\tau,\tau+1)$.  It is also a $(K_0+1)\varepsilon$-net for $\varepsilon $-conservative words.  
\end{restatable}


\begin{restatable}[$\varepsilon$-net for $L_T(\aut)$]{con}{epsNetAut}
\label{con:net}
Given $T$ we proceed as follows:
\begin{itemize}
    \item for each $n\leq f_\aut(\varepsilon)$, take all possible sequences $\zeta$ of breakpoints in $\varepsilon\nat$:
\[
0=a_0\leq b_0\leq a_1\leq b_1\cdots\leq a_n\leq b_n \leq T
\]
 with cumulative duration of odd intervals $\sum_i (b_i-a_i)\leq f_\aut(\varepsilon)$;
 \item for each such sequence $\zeta$ we construct an $\varepsilon$-net $\Net^\zeta_\varepsilon$ like that:
 \begin{itemize}
     \item on each odd interval $[a_i,b_i]$ put the universal $\varepsilon$-net for duration $b_i-a_i$;
     \item on each even interval $[b_i,a_{i+1}]$  put 
$
\Net^{\mathbf{cons}}_{\varepsilon/(K_0+1),a_{i+1}-b_i}(b_i-1/2);
$
 \end{itemize}
  \item let $\Net_{\varepsilon,T}=\bigcup_\zeta \Net^\zeta_\varepsilon$ with the union taken over all breakpoint sequences.
\end{itemize}
\end{restatable}


\begin{restatable}{lemma}{lemmaIndeedEpsNetAut}
\label{lem:it:is:a:net}
The set $\Net_{\varepsilon,T}$ described in \cref{con:net} is an $\varepsilon$-net for $L_T(\aut)$.
\end{restatable}
Finally we estimate the cardinality of those nets.
\begin{restatable}{proposition}{propEpsCount}
\label{lem:eps:count} The cardinalities of the sets defined in the previous constructions  can be bounded as follows (with $K_i$ depending only on the automaton $\aut$):
\begin{align}
 \log\#\Net^{\mathbf{fol}}_\varepsilon(\pi,\pi',\tau) &\leq (\alpha c(\pi')+K_3)\log\frac1\varepsilon+ K_5c(\pi')+K_{12}; \label{eq:count:fol}
 \\
 \log\#\Net^{\mathbf{cons}}_{\varepsilon,\Theta}(\tau)&\leq \alpha \Theta \log\frac1\varepsilon + K_3\log\frac1\varepsilon + K_6 \Theta +K_7; \label{eq:count:cons}
 \\
\limsup_{T\to\infty}
\frac{\log\#\Net_{\varepsilon,T}}{T} &\leq  \alpha \log\frac1\varepsilon +K_9. \label{eq:count:all}
\end{align}
\end{restatable}
We then obtain the upper bound for $\bandh_\varepsilon$:
\begin{restatable}[Upper bound]{proposition}{propUpperBound}
\label{lem:upper}If $\alpha>0$, then exists a constant $K_9$ such that for all $\varepsilon>0$  it holds that
$$\bandh_\varepsilon(L(\aut))\leq \alpha\log\frac1\varepsilon +K_9.$$
\end{restatable}

The final step is given by the following proposition which also complements the results from \cite{3classes}:
\begin{restatable}{proposition}{propAlphaZero}
\label{lem:meager}
If the graph $G$ is acyclic or $\alpha=0$, then $\aut$ is meager.
\end{restatable}

Now we may combine \cref{lem:upper} with \cref{lem:lower} and the inequalities  \eqref{eq:band:ineq} to get the identities in \cref{thm1}.


%% file: conclusion.tex
\subsection{Summarizing four papers}
Putting together the results from \cite{3classes,CIAA,obese} and this article, we obtain the following comprehensive result:
\begin{thm}\label{thm:all}
Given a nondeterministic timed automaton $\aut$ without $\epsilon$-transitions, one of the following holds for its language $L$ (for $\varepsilon\to 0$):
\begin{itemize} 
\item $L$ is \textbf{meager}, i.e. $\bandc_\varepsilon (L)=O(1)$ and 
$\bandh_\varepsilon (L)=O(1)$. If moreover $\aut$ is deterministic, a number $\alpha\geq 0$ can be computed (as a logarithm of an algebraic number), such that both $\bandc_\varepsilon (L)$ and $\bandh_\varepsilon (L)$  equal $\alpha+o(1)$. 
\item $L$ is \textbf{normal}: a rational $\alpha>0$ can be computed (in \FPSPACE), such that both $\bandc_\varepsilon (L)$ and $\bandh_\varepsilon (L)$ equal $\alpha\log\frac1 \varepsilon+o\left(\log\frac1\varepsilon\right)$. 
\item $L$ is \textbf{obese}: $\alpha>0$ can be computed (as a logarithm of an algebraic number), such that both $\bandc_\varepsilon (L)$ and $\bandh_{\varepsilon/2} (L)$ equal $\frac\alpha\varepsilon +o\left(\frac1\varepsilon\right)$.
\end{itemize}
Given $\aut$, the three decision problems ``is it meager?'', ``is it normal?'', ``is it obese?'' are \PSPACE-complete.
\end{thm}

\subsection{Perspectives}
While \Cref{thm:all} is almost exhaustive, a couple of questions should still be explored:
\begin{itemize}
\item settling the exact complexity of computing $\alpha$  (exactly and approximately) in each of the three cases;
\item computing $\alpha$ for non-deterministic meager automaton;
\item and analyzing the bandwidth for timed automata with $\epsilon$-transitions.
\end{itemize}
More generally, we are also interested in: 
\begin{itemize}
    \item more precise computation of $\capa_\varepsilon(L_T),\ent_\varepsilon(L_T)$, in the spirit of \cite{generating}, where not only the asymptotic exponential growth rate, but also the exact generating function was computed (for a different ``metrics''); 
    \item interpretation of bandwidth  in terms of Kolmogorov's complexity and topological entropy;
    \item timed coding/decoding by transducers;
    \item applications to information transmission and compression of timed data.
\end{itemize}

Our results and techniques strongly depend on specific properties of timed automata. More general study of extended automata, especially hybrid ones, from the standpoint of information theory,  remains an attractive research subject. 

%% file: appendix.tex
\section{Preliminaries}

\subsection{On geometry of points, runs and words}\label{sec:geom}
\emph{A busy reader can skip this section except its first paragraph, \cref{def:poly:lang,lem:sep} that will be used in the proof of the lower bound. A motivated reader is invited to revise \cref{sec:pseudo} before continuing.}

We recall some simple geometric  notions and facts. For a vector $\z\in\real^n$, its \emph{$\infty$-norm} is $||\z||_\infty=\max_{i\in1..n}|z_i|$, and the associated \emph{$\infty$-distance} between $\y$ and $\z$ is $||\y-\z||_\infty$. A (convex) \emph{polytope}  $P\subset\real^n$ is the set of vectors satisfying a finite system of linear constraints; it has a \emph{dimension}  $\geq k$ whenever there exists a point  $\x\in P$ and $k$ linearly independent vectors  $\xi_j\in\real^n, j=1..k$, such that $\x+\xi_j\in P$ for all $j=1..k$. 
\begin{lemma}[folk]\label{lem:sep:infty}
For any polytope $P$ of dimension $\geq k$, there exists $\beta>0$, such that for  all $\varepsilon$ small enough $P$ contains an $\varepsilon$-separated set (with respect to $\infty$-distance) of cardinality $>\beta/\varepsilon^k$. 
\end{lemma}
In other words, with respect to $\infty$-distance, $\capa_\varepsilon(P)\geq k\log(1/\varepsilon)+O(1)$.

%
%

We will apply geometric reasoning to runs and timed words; to this aim, we associate vectors (denoted with the same symbols in boldface) with them. 
We associate with each run 
 $$\rho = (q_0,\x_0) \trans{(\delta_1,d_1)} (q_1,\x_1)\cdots \trans{(\delta_n, d_n)}(q_n,\x_n),
 $$
 a vector $\vec{\rho}=(\x_0,\x_1,\dots,\x_n)\in \real^{(n+1)\#X}$.
We associate with each timed word $w=(a_1,t_1)\dots(a_n,t_n)$  its timing $\w=(t_1,\dots t_n) \in \Kone_n \subset\real^n$ and untiming $\untiming(w)=a_1\dots a_n \in \Sigma^*$. Here, the cone $\Kone_n$ of all possible timings is defined by inequalities $0\leq t_1\leq\cdots t_n$. 

Our next aim is to port \cref{lem:sep:infty} to timed words, which is not straightforward, since pseudo-distance $d$ is quite different from $\infty$-distance. 

We define the function $\gap:\Kone_n\to \real_+$ as follows : $\gap(t_1,\dots,t_n)=\min\{t_{i+1}-t_i | i\in 1..n-1, t_{i+1}>t_i\}$. Given a polytope $P\subseteq \Kone_n$ let $\gap(P)=\inf_{\z\in P}\gap(\z)$. In words, $\gap(P)=\mu$ means that adjacent events in any timing in $P$ are either simultaneous, or separated by at least $\mu$ time units.

\begin{definition}\label{def:poly:lang}
Given a polytope $P\subseteq \Kone_n$ with $\dim P = k$, we say a timed language $L$ is \emph{polytopic} (of dimension $k$) with respect to $P$ whenever all words of $L$ have the same untiming and the set of timings of $L$ is exactly $P$.
\end{definition}
For instance, we remark that the language of any single path of a TA is polytopic, and its dimension is the number of non-punctual transitions of the path.

\begin{lemma}\label{lem:gap}Every polytope $P\subseteq \Kone_n$ of dimension $k$ contains a polytope $P'\subseteq P$ of the same dimension and with $\gap(P')>0$. 
\end{lemma}
\begin{proof}
We choose an interior point $\x=(x_1,\dots,x_n)\in P$ and $k$ linearly independent vectors  $\xi_j$, such that $\x+\xi_j\in P$ for all $j=1..k$. 
Let $I\subseteq 1..n-1$ contain all the indices $i$ for which $x_{i+1}>x_i$. 
We define $\mu=\min_{i\in I} (x_{i+1}-x_i)/2$ and $P'$ the subset (polytope) of $P$ containing the points $\y$ satisfying $\forall i\in I, y_{i+1}-y_i\geq\mu$. The polytope $P'$ obtained in this way has a gap $\mu$; also, $\x\in P'$ and $\x+\xi_j/2\in P'$, thus the dimension is $\geq k$. \qed
\end{proof}

The notion of gap helps relate the pseudo-distance $d$ with the $\infty$-distance. 

\begin{lemma}\label{lem:compare}
Let $w^1$ and $w^2$ be timed words with the same untiming. Then:
\begin{itemize}
\item $d(w^1,w^2)\leq ||\w^1-\w^2||_\infty$;
\item if moreover $w^1,w^2$ belong to  the same polytopic language with polytope $P$ of gap $\mu > 0$, and  $d(w^1,w^2)<\mu/2$, then $d(w^1,w^2)= ||\w^1-\w^2||_\infty$.  
\end{itemize}
\end{lemma}

%
%
\begin{proof} We denote $||\w_1-\w_2||_\infty $ by $D$.
For each letter $(a_i,t^1_i)$ in $w^1$, its distance from the corresponding letter $(a_i,t^2_i)$ in $w^2$ does not exceed $D$; thus $\overrightarrow d(w_1,w_2)\leq D$. Similarly, $\overrightarrow d(w_2,w_1)\leq D$ and the first statement is proved.

Suppose now that the hypotheses of the second statement are fulfilled.  Let us prove first that $||\w_1-\w_2||_\infty <\mu/2$, that is
\begin{equation}\label{eq:ind}
    \forall i,\, |t^1_i-t^2_i|<\mu/2 
\end{equation}
by induction over $i$. 

W.l.o.g., suppose that $w^1$ starts earlier than $w^2$ (i.e.~$t^1_1\leq t^2_1$). In this case, the closest matching letter for $(a_1,t^1_1)$ is $(a_1,t^2_1)$ and thus $|t^1_1-t^2_1|\leq \overrightarrow d(w_1, w_2) <\mu/2 $, and the basis of induction is established.

Suppose now, by inductive hypothesis, that \eqref{eq:ind} holds for some $i$; we want to prove it for $i+1$. There are two cases, by definition of the gap:
\begin{itemize}
    \item in polytope $P$, for all points $t_{i+1}=t_i$. In this case, $|t^1_{i+1}-t^2_{i+1}|=|t^1_{i}-t^2_{i}| <\mu/2$, as required.  
    \item in polytope $P$, for all points $t_{i+1}-t_i\geq\mu$. 
    Suppose also that $t^1_{i+1}< t^2_{i+1}$ (the symmetric case is similar). The closest matching letter for $(a_{i+1},t^1_{i+1})$  is some $(a_{i+1},t^2_{j})$ in $w^2$ (if there are several equivalent matches, we take the smallest index), and it should satisfy $|t^1_{i+1}-t^2_{j}|\leq d(w^1,w^2)<\mu/2$. 
    
    
    {There are three cases. 
    \begin{itemize}
        \item $j=i+1$ and we are done: $|t^1_{i+1}-t^2_{i+1}|<\mu/2$ and \eqref{eq:ind} holds for $i+1$;
        \item $j\leq i$, in this case
        $$
        \begin{cases}
        t^1_i>t^2_i-\mu/2 &\text{by inductive hypothesis};\\
        t^1_{i+1}<t^2_j+\mu/2<t^2_i+\mu/2 & \text{by lemma hypothesis, since } j\leq i;\\
        t^1_{i+1}\geq t^1_i+\mu & \text{because of the gap},
        \end{cases}
        $$
        which is contradictory.
        \item $j\geq i+2$. However, letter $(a_{i+1},t^2_{i+1})$ stands between $t^1_{i+1}$ and $t^2_j$ and is thus a better match for $(a_{i+1},t^1_{i+1})$, which contradicts the choice of $j$.
    \end{itemize}
 We conclude that in the only possible case, the inductive step is valid. } 
\end{itemize}
We have proved that $D=||\w_1-\w_2||_\infty <\mu/2$. 



This implies that for any letter $(a_{i},t^1_{i})$ in $w^1$, the best match in $w^2$ is  $(a_{i},t^2_{i})$. Indeed, it is at the distance $|t^1_{i}-t^2_i|\leq D<\mu/2$. On the other hand, for all $j\neq i$ (except the trivial case when in $P$ all points satisfy $t_j=t_i$), by gap condition $|t^1_j - t^1_i|\geq\mu$,  hence $|t^2_j - t^1_i|\geq \mu/2$. 


So under the hypothesis of the second bullet of the lemma, $\overrightarrow d(w_1, w_2)$ simplifies as $\max_i |t^2_i - t^1_i|$, which is exactly $||\w_1-\w_2||_\infty$. The same is true for $\overrightarrow d(w_2, w_2)$ and hence for $d(w_1, w_2)$. 
\qed
\end{proof}

\begin{lemma}\label{lem:sep}
Let $L$ be a polytopic timed language of dimension $k$. Then there exists $\beta>0$ and $\mu>0$ such that for any $0<\varepsilon<\mu/2$, $L$ contains an $\varepsilon$-separated set (w.r.t the pseudo-distance $d$) of cardinality $>\beta/{\varepsilon}^k$.
\end{lemma}

\begin{proof}
We first apply \cref{lem:gap} to the polytope $P$ of $L$, and obtain $P'\subseteq P$ of dimension  $k$ with $\mu=\gap(P')>0$. Let now $\varepsilon<\mu/2$ and $\Sep\subset P'$ an  $\varepsilon$-separated set of cardinality $>\beta/{\varepsilon}^k$ in $P'$  w.r.t $||\cdot||_\infty$ (existing by \cref{lem:sep:infty}). 
For any two distinct words $w_1,w_2\in L$ with timings in $\Sep$,
\begin{itemize}
    \item either $d(w_1,w_2)\geq \mu/2>\varepsilon$,
    \item or $d(w_1,w_2)< \mu/2$, and by \cref{lem:compare},  $d(w_1,w_2)=||\w^1-\w^2||_\infty>\varepsilon$. \qed
\end{itemize}
\end{proof}

\subsection{Exercises in Puri's theory}
Definitions and results in this section are citations, variants, or extensions of those in the groundbreaking \cite{puri}, adapted to our settings. Remark that  these results are stated for the closed semantics of \RTA. 


Let us start with a convexity result. 
\begin{lemma}[\cite{puri}]\label{lem:puri:conv}
  In a TA, whenever $\rho_1,\dots,\rho_m$ are runs along the same path, and $\lambda_i\geq0, \ i=1..m$ satisfy $\sum \lambda_i=1$, then there exists a run $\rho$ along the same path with $\vec{\rho}=\sum \lambda_i\vec{\rho}_i$. This also holds under the closed semantics. 
\end{lemma}
We will call such a run $\rho$ a convex combination of $\rho_i$. 


\defOrbitMorphism*


For the empty path $\epsilon$, we define  $\gamma(\epsilon)=1$;  for a sequence of edges $\pi$ that is not a path (i.e.~is disconnected), we put  $\gamma(\pi)=0$, with $0$ and $1$ special elements. 

 First, let us define the operation $\circ$ representing graph composition: $G_1\circ G_2$ is the graph with edges  $(v, v'')$ such that there exists $v'$ with $(v,v')$ an edge of $G_1$ an $(v', v'')$, an edge of $G_2$. Consider now the set $O$ of the orbits of all edge sequences of a \RTA\ and let us define the operation $\cdot$ on $O$, such that $\langle p_1, G_1, p_2\rangle\cdot\langle p_2, G_2, p_3\rangle = \langle p_1, G_1\circ G_2, p_3\rangle$ (graphs do compose when orbits agree on their junction location), $1\cdot e = e \cdot 1 = e$ holds for any $e$ ($1$ is neutral) and  $e_1\cdot e_2 = 0$ holds in all the other cases. It is easy to see that $(O,\cdot,1)$ is a finite monoid, and $\gamma:\Delta^*\to O$ as in \cref{def:orbit:long} is a monoid morphism. In the sequel, when source and target locations are implied by context, we will just identify orbits with their graph component and describe them using graph terminology.

It is often useful to describe points within a simplex (e.g.~a closure of a region) using the following.
\begin{definition}[Barycentric coordinates, \cite{bary}]\label{def:bary}
Given a simplex $S$ with a set of vertices $V$ and a vector $\x\in S$, there is a unique vector $\lambda^S(\x)\geq 0$ such that $\x=\sum_{v\in V} \lambda^S_v(x) v$ and $\sum\lambda^S_v(\x)=1$.
We call $\lambda^S(\x)$ the \emph{barycentric coordinates} of $\x$ in $\bar S$.
\end{definition}
 We usually omit the superscript $S$ when the simplex is clear from context.

Reachability in \RTA\ can be characterized in terms of barycentric coordinates and orbit graphs, this is the key result of  Puri's theory. 
\begin{lemma}[\cite{puri}]\label{lem:puri:reach}
Given a path $\pi$ from $p$ to $q$ in an \RTA\ with orbit graph $e=\gamma(\pi)$, let $\{v_i\}$ be the vertices of $S(p)$ and $w_j$ the vertices $S(q)$,  and let $\x=\sum_i \lambda_i v_i\in \bar{S}(p)$ and $\y=\sum_j\mu_jw_j\in \bar{S}(q)$ two clock valuations and their barycentric coordinates. Then $(q,\y)$ is reachable from $(p,\x)$ along the path $\pi$ (by the closed semantics) iff there exists a non-negative matrix $P=(p_{ij})$ which 
\begin{itemize}
    \item is dominated by $e$, i.e.~$p_{ij}>0$ only if $e$ contains an edge $(v_i,w_j)$;
    \item is stochastic, i.e.~satisfies $\sum_j p_{ij}=1$ for all $i$;
    \item transforms $\x$ in $\y$, i.e.~satisfies $\sum_i p_{ij}\lambda_i =\mu_j$ for all $j$.
\end{itemize}
\end{lemma}
We will extensively use the following consequence of the previous result.   
\begin{lemma}\label{lem:puri:proj}
 In an \RTA\ $\aut$, we fix a path $\pi$, going from location $p$ to location $q$, with orbit graph $e=\gamma(\pi)$ and denote the vertices of $S(p)$ as $\{v_i\}$ and those of $S(q)$ as $\{w_j\}$. 
 
Let $\rho\in\overline{\Runs_\pi}(\x,\y)$, then $\vec\rho$ can be represented as
$\sum_{ij}\alpha_{ij}\vec\rho_{ij}$ with 
\begin{itemize}
    \item $\rho_{ij}\in\overline{\Runs_\pi}(v_i,w_j)$ (defined when $\alpha_{ij}\neq 0$);
    \item $\alpha_{ij}\geq 0$ and $\sum_{ij} \alpha_{ij}=1$;
    \item coefficients $\alpha_{ij}$ dominated by the orbit graph $e$, i.e.~$\alpha_{ij}>0$ only if $e$ contains an edge $(v_i,w_j)$.
\end{itemize}
\end{lemma}
\begin{proof}
In the proof below we explicitly define coefficients $\alpha_{ij}$ and their domination by $e$ is straightforward. 

We prove the lemma by induction over the length of $\pi$. For the \textbf{base case}, i.e.~one transition $\rho=(p,\x)\to (q,\y)$, the result  follows from \cref{lem:puri:reach}. 
Indeed, take
 $\rho_{ij}$ being just an edge $(p,v_i)\to (q,u_j)$, and $\alpha_{ij}=p_{ij}\lambda_i$. We have that
$$\sum_{ij}\alpha_{ij}=\sum_{ij}p_{ij}\lambda_i=\sum_j \sum_ip_{ij}\lambda_i=\sum_j\mu_j=1,$$
as required. Also,
$$\sum_{ij}\alpha_{ij}\vec\rho_{ij}=
\sum_{ij}p_{ij}\lambda_i[v_i|u_j],
$$
where $[v|u]$ stands for concatenation of two vectors.
Let us check that the last sum coincides with $\rho=[\x|\y]$. For the first half of coordinates:
$$\sum_{ij}p_{ij}\lambda_i v_i= 
\sum_{i}\left(\sum_j p_{ij}\right)\lambda_i v_i=
\sum_{i}\lambda_i v_i=\x,
$$
since by stochasticity $\sum_j p_{ij}=1$. As for the second half,
$$\sum_{ij}p_{ij}\lambda_i u_j= 
\sum_{j}\left(\sum_i p_{ij}\lambda_i\right) u_j=
\sum_{j}\mu_j u_j=\y,
$$
and the base case is done: $\sum_{ij}\alpha_{ij}\vec\rho_{ij}=[\x|\y]=\vec\rho$.

For  the \textbf{inductive hypothesis}, let $\rho$ from $(p,\x)$ to $(q,\y)$ (of some length $n$) be decomposed as required into runs $(p,v_i)\Trans{\rho_{ij}}(q,w_j)$. Let again $\mu_j$ be the barycentric coordinates of $\y$.
We have that 
$$
\y=\sum_j\mu_j w_j;\quad \vec\rho=\sum_{ij}\alpha_{ij}\vec\rho_{ij};\quad \sum_{ij}\alpha_{ij}=1. 
$$
Projecting the second equality onto the coordinates of the last clock vector in the run we get $\y=\sum_{ij}\alpha_{ij}w_{j}=\sum_j \left(\sum_i \alpha_{ij}\right)w_j,$
and thus $\sum_i \alpha_{ij}$ coincides with barycentric coordinates of $\y$:
$$
  \sum_i \alpha_{ij}=\mu_j.  
$$
For the \textbf{inductive step}, consider a longer run 
$$
\rho^+: (p,\x) \Trans{\rho} (q,\y) \trans{\delta} (r,\z), 
$$
with $u_k$ vertices of $S(r)$ and $\z=\sum_k \nu_ku_k$. By \cref{lem:puri:reach} applied to $\delta$, there exists a stochastic matrix $P=(p_{jk})$ dominated by $\gamma(\delta)$, such that $\nu_k=\sum_j\mu_jp_{jk}$.

We will now decompose $\rho^+$ as required. Let
$$
\alpha'_{ik}=\sum_j \alpha_{ij}p_{jk};\quad
\vec\rho'_{ik}=\frac{1}{\alpha'_{ik}}\sum_j \alpha_{ij}p_{jk}[\vec\rho_{ij}|u_k];\quad
\vec\rho'=\sum_{ik}\alpha'_{ik}\vec\rho'_{ik}.
$$
   We will prove that $\rho^+=\rho'$ and all the assertions of the lemma are satisfied.

\begin{itemize}
\item \textbf{Sum of $\alpha'$.} Let us check that it equals $1$. 
$$
\sum_{ik}\alpha'_{ik}=\sum_{ikj} \alpha_{ij}p_{jk}=
\sum_{ij}\alpha_{ij}\left(\sum_k p_{jk}\right)=\sum_{ij}\alpha_{ij}=1,
$$
since $\sum_k p_{jk}=1$  by stochasticity of $P$. 
\item \textbf{$\rho'_{ik}$ are runs.}  We observe  that 
 $[\rho_{ij}|u_k]$ is a run (if $\alpha_{ij}>0$   and $(w_j,u_k)\in \gamma(\delta))$. Thus $\vec\rho'_{ik}$ is a linear combination of runs $[\vec\rho_{ij}|u_k]$, we will check that it is a convex one (this would imply that $\rho'$ is a run by \cref{lem:puri:conv}). Indeed, the sum of the coefficients of $[\vec\rho_{ij}|u_k]$ equals $1$:
 $$
 \sum_{j}\frac{1}{\alpha'_{ik}}\alpha_{ij}p_{jk}=\frac{1}{\alpha'_{ik}}\sum_{j}\alpha_{ij}p_{jk}= \frac{\alpha'_{ik}}{\alpha'_{ik}}=1,
 $$
 as required. 
 
 To prove that $\rho'=\rho^+$, we deal separately with the projection on the first $n$ states of the runs  $\mathbf{pref}(\cdot)$ and on the last state $\mathbf{last}(\cdot)$.
 
\item \textbf{Prefix}. By inductive hypothesis,
\begin{multline*}
\mathbf{pref}(\rho')= 
\sum_{ik}\alpha'_{ik}\mathbf{pref}(\rho'_{ik})=\sum_{ik}\alpha'_{ik}\frac{1}{\alpha'_{ik}}\sum_j \alpha_{ij}p_{jk}\vec\rho_{ij}=\sum_{ijk}\alpha_{ij}p_{jk}\vec\rho_{ij}=\\
\sum_{ij}\alpha_{ij}\left(\sum_k p_{jk}\right)\vec\rho_{ij}=\sum_{ij}\alpha_{ij}\vec\rho_{ij}=\vec\rho=\mathbf{pref}(\rho^+),
\end{multline*}
 as required.
 \item \textbf{Last state.} We proceed with a final check:  
 \begin{multline*}
 \mathbf{last}(\rho')= 
\sum_{ik}\alpha'_{ik}\mathbf{last}(\rho'_{ik})=\sum_{ik}\alpha'_{ik}\frac{1}{\alpha'_{ik}}\sum_j \alpha_{ij}p_{jk}u_k=\\
\sum_{ijk}\alpha_{ij}p_{jk}u_k=
\sum_{jk}\left(\sum_i\alpha_{ij}\right)p_{jk}u_k=
\sum_{jk}\mu_j p_{jk}u_k=\\ \sum_k\left(\sum_j\mu_jp_{jk}\right)u_k=\sum_k \nu_k u_k=\z=\mathbf{last}(\rho^+),
\end{multline*}
and thus $\vec\rho^+=\vec\rho'$ and it can be decomposed as $\sum_{ik}\alpha'_{ik}\vec\rho'_{ik}$. This concludes the induction. 
\end{itemize}
\end{proof}

Similarly to \cite{entroJourn,3classes}, we will use another consequence of \cref{lem:puri:reach}: a Lyapunov function for the return map of runs along a cyclic path. 
\begin{lemma}[\protect{\cite[Lemma 11]{entroJourn}}] \label{lem:puri:Lyap}
Given a cyclic path $\pi$ and $I$, an inital (i.e.~with no incoming edges) subset of the vertices of $\gamma(\pi)$, then $(\x,\x')\in \overline{\mathop{Reach}_\pi}$ implies $\sum_{v\in I}\lambda_v(\x)\geq \sum_{v\in I}\lambda_v(\x')$.
\end{lemma}
In other words,  the function $\x \mapsto \sum_{v\in I} \lambda_v(\x)$ is a Lyapunov function (for the return map of $\pi$).




\subsection{Some properties of heartbeat}
Heartbeat  brings a couple of useful properties. 
\begin{lemma}\label{lem:heart} A TA with a heartbeat has the following properties: 
\begin{itemize}
\item the duration of a run with $k$ heartbeats (events $b$)  from a state with $h=h_s$ to a state with $h=h_e$ equals $k-h_s+h_e$ and thus belongs to the interval $[k-1,k+1]$;
\item 
let $\rho_1,\rho_2$ be two runs along the same path, with timed words $w_1,w_2$. It holds that
$$d(w_1,w_2)\leq 2||\vec{\rho}_1-\vec{\rho}_2||_\infty.$$
\end{itemize}
\end{lemma}
We remark that the latter property does not hold without a heartbeat. 
\begin{proof}
To prove the former property, we remark that the run spends $1-h_s$ time before the first beat, $k-1$ time units from the first to the last beat, and $h_e$ time units after the last beat. 

To prove the latter, let the norm in the right-hand side be $\nu$. Consider the $i$th position in the path, a letter $(a_i,t_{1i})$ (in the word $w_1$) and a letter $(w_{2i},t_{2i})$  in $w_2$, 
they  correspond respectively to transitions $$(q_{i-1},\x_{1,i-1}) \trans[d_{1i}]{\delta}(q_i,\x_{1,i}) \text{ and }(q_{i-1},\x_{2,i-1}) \trans[d_{2i}]{\delta}(q_i,\x_{2,i});
$$
due to the first assertion of the lemma $t_{1i}=k-h_{10}+{h_{1i}}$ and $t_{2i}=k-h_{20}+h_{2i}$. By definition $h_{10}$ is a coordinate of $\x_{10}$ etc, hence $|h_{10}-h_{20}|\leq\nu$ and $|h_{1i}-h_{2i}|\leq\nu$, and we conclude by the triangular inequality that $|t_{1i}-t_{2i}|\leq 2\nu$. Since this holds for any letter of $w_1$ and symmetrically $w_2$, we deduce that $d(w_1,w_2)\leq 2\nu$ as required. \qed
\end{proof}

\section{Lower bound}
%
\subsection{Case of finite $\alpha$}
\propLowerBound*



We recall the setting  for proving proposition \cref{lem:lower}:
%
\begin{equation}
\tag{\ref{eq:set}}
\text{\parbox{0.93\textwidth}{
\begin{itemize}
    \item $\pi,\pi'$ is a cyclic route from $(q,\ell)$ to itself with $\ell$ containing $k$ elements; 
    \item  $\sigma$  is a cyclic path from $q$ to $q$ in $\aut$ such that all $\ell_i$ are cliques in $\gamma(\sigma)$.
\end{itemize}
}}    
\end{equation}
We recall that (by \cref{def:facelist,con:main}) such a  ``resetting'' cycle $\sigma$ exists for each vertex $(q,\ell)$ in $G$. 

\lemSepOne*

The proof is based on the techniques of \cref{sec:geom}.
\begin{proof}
Let $R\subseteq S(q)$ be defined%
\footnote{Geometrically, for $k>1$, the set $R$ is a cross-section of the simplex $S(q)$ parallel to all the faces $\conv\ell_i$.} 
as $\frac1k\sum_{i=1}^k \conv\ell_i$.  Let $\x^*_i$ be the barycenter of $\ell_i$ and  $\x^*=\frac1k\sum_{i=1}^k  \x^*_i$.

Let $P_i$ be the set of timings of traces of  all runs in  $\overline{\Runs_\pi}$ starting in  $(q,\x^*_i)$  and hitting at the end the face $\conv\ell_i$. 
Finally, let $P=\frac1k\sum_{i=1}^k  P_i$. We make several observations concerning $P$:
\begin{itemize}
    \item its vectors are timings of traces of runs along $\pi$ from $(q,\x^*)$ to some state in $(q,R)$;
    \item they are bounded by  $t'$ since they  correspond to runs with $c(\pi)$ heartbeats (by \cref{lem:heart}). 
    \item $P$ is a convex polytope, and it has a dimension of at least $r(\pi)$. Indeed, the definition of a vector in $P_i$ can be stated as a conjunction of linear inequalities over timestamps in the trace; thus, $P_i$ is a convex polytope. Hence, $P$ is also a convex polytope. To estimate its dimension, we  find some vector  $\vec{v}\in P $ and $r(\pi)$ linearly independent vectors $\xi^j$ such that $\vec{v}+\xi^j\in P$, proceeding as follows: 
    \begin{itemize}
    \item  for each $i$ we take $\vec{v}_i$ timing of a run from $(q,\x^*_i)$ through internal points of all active faces $\conv\ell_i$ along the path;  such a run exists by virtue of \cref{lem:puri:reach}.  We have that $\vec{v_i}\in P_i$ 
        \item consider now $\vec{v}=\frac1k\sum_{i=1}^k \vec{v_i}\in P$, it also corresponds to a run along $\pi$;
        \item there are $r(\pi)$ edges $\delta^j=(q^{j-1},\ell) \to (q^j,\ell'^{j})$ on $\pi'$ having $r(\delta^j)=1$. For each of those, we take the smallest index $i(j)$ such that for some $v\in\ell_i^{j-1}$ the vertex $v$ can make a transition to $(q,\ell_i^{j})$ after any time in some interval $(l,u)$. Hence, $\vec{v}_i$  can also take the transition a bit earlier or later. Let $\vec{v}'^j_{i(j)}\in P_{i(j)}$ be the vector of timings of a run that agrees with $\vec{v}_i$ up to transition $\delta^{j-1}$, takes $\delta^j$ a bit earlier or later, and is completed to the full length of the path $\pi$. Let $\vec{v}'^j=\frac1k \left(\vec{v}'_{i(j)}+\sum_{i\neq i(j)} \vec{v}_i\right)$ and $\xi^j=\vec{v}'^j-\vec{v}$.
        \item By construction $\vec{v}+\xi^j\in P$.
        \item The  non-zero coordinate of $\xi^j$ with the smallest index is $t_j$, hence all these vectors are linearly independent.
    \end{itemize}
    We conclude that the dimension of $P$ is at least $r(\pi)$.
    \item Hence, \cref{lem:sep} applies and  for $\varepsilon$ small enough one can find an $\varepsilon$-separated set $\Sep$ within the timed words having timing in $P$, such that  $\#\Sep\geq\beta/\varepsilon^{ r(\pi)}$, for some $\beta>0$.
\end{itemize}
\end{proof}

\lemReturn*


\begin{proof}
By the theory of \cite{puri}, it follows from the condition on $\gamma(\sigma)$ that within  $\conv\ell_i$ every point is reachable from every point by a run along $\sigma$.

Let $\y_i\in\conv\ell_i$ be such that $\y=\frac1k\sum \y_i$. Let $\rho_i$ be a run from $\y_i$ to $\x^*_i$. We have that $\vec\rho^\y=\frac1k\sum  \vec\rho_i$ is a run from $\y$ to $\x^*$, as required. The duration bound follows from \cref{lem:heart}.  
\end{proof}
\corSigmaPi*

We remark that words in $\Sep'$ have runs cycling on $(q,\x^*)$ and can be concatenated to each other within $L_{(\pi\sigma)^*}(\x^*,\x^*)$.
\begin{proof}
Every word in $\Sep$ (from \cref{lem:sep:one}) has a run along $\pi$ from $(q,\x^*)$ to some state  $(q,\y)$ with $\y\in R$, and we can append to it a run $\rho^\y$ (from \cref{lem:return}) along $\sigma$ leading back to $(q,\x^*)$ and of duration $\leq \beta(\sigma)+1$. Thus we obtain an $\varepsilon$-separated  set $\Sep'\subset L_{\pi\sigma}(\aut)$ of cardinality  $c/\varepsilon^{ r(\pi)}$ and duration at most $t''$.  \qed
\end{proof}

We are now ready to prove the lower bound.
 
\begin{proof}[of \cref{lem:lower}] 
Let $\pi$ be an optimal cycle in the refined graph G from $(q,\ell)$ to itself, with  $r(\pi)=\alpha c(\pi)$. Let $\sigma$ be a resetting cycle for $(q,\ell)$. 

 Intuitively, most of the time  $\pi$ will be used for coding, and when needed, $\sigma$ will be taken to reset clocks.

Let us fix $\varkappa>0$, set $\theta=\left\lceil \frac{2\alpha(c(\sigma)+2)}{\varkappa c(\pi)}\right\rceil$ and  apply \cref{cor:one} to the cycle $\pi^\theta$. This yields an $\varepsilon$-separated set $\Sep'$ in $L_{\pi^\theta\sigma, t^*} (x^*,x^*)$, with $t^*= \theta c(\pi)+c(\sigma)+2 $ of cardinality $\geq  \beta /\varepsilon^{ar(\pi) }$.

Since runs corresponding to words in $\Sep'$ cycle from $(q,x^*)$ to the same state, they can be iterated indefinitely. Let also  timed words $u$ and $v$ be traces of runs from an initial state to $(q,x^*)$ and from $(q,x^*)$ to a final state, respectively and $\nu=\length(uv)$.  Thus $\Sep_T=u\Sep'^{\lfloor (T-\nu)/t^*\rfloor} v $ is an $\varepsilon$-separated set within $L_T(\aut)$. It provides a lower bound on the bandwidth 
$$
\bandc_\varepsilon(L)\geq \limsup_{T\to \infty} \frac{\log\#\Sep_T}{T}.
$$
The right-hand side equals 
\begin{multline*}
\limsup_{T\to \infty} \frac{\left\lfloor \frac{T-\nu}{t^*}\right\rfloor \log\#\Sep'}{T}\geq
\limsup_{T\to \infty} \frac{\left(\frac{T-\nu}{t^*} -1\right)(\log\beta+ar(\pi)\log(1/\varepsilon))}{T}=\\
\frac{\log\beta+\theta\alpha c(\pi)\log(1/\varepsilon)}{\theta c(\pi)+c(\sigma)+2}=\log(1/\varepsilon)\cdot (A_1+A_2),
\end{multline*}
with 
$$
A_1=\frac{\log\beta}{\log(1/\varepsilon)\cdot(\theta c(\pi)+c(\sigma)+2)}\geq -\varkappa/2
$$
for $\varepsilon$ small enough,
and 
$$
A_2= \frac{\alpha \theta c(\pi)}{\theta c(\pi)+c(\sigma)+2}= \frac{\alpha}{1+\frac{c(\sigma)+2}{\theta c(\pi)}}\geq \alpha\left(1-\frac{ c(\sigma)+2}{\theta c(\pi)}\right)\geq \alpha-\varkappa/2. 
$$
From the four inequalities above, we conclude that 
 $\bandc_\varepsilon(L(\aut))\geq(\alpha-\varkappa)\log(1/\varepsilon)$ as required. \qed
\end{proof}
\subsection{The extreme case}
\propAlphaInfty*
\begin{proof}
If $\alpha =\infty$, then $G$ contains a cycle $\pi$ from $(q,\ell)$ to itself, with cost $c(\pi)=0$ (that is, without heartbeat events) and a positive reward $r(\pi)\geq 1$. Let $\sigma$ be a resetting cycle for $(q,\ell)$.  We fix an $A>0$, and proceed similarly to the proof of \cref{lem:lower}. 
We set $\theta=A(c(\sigma)+2)$ and  apply \cref{cor:one} to the cycle $\pi^\theta$. This yields an $\varepsilon$-separated set $\Sep'$ in $L_{\pi^\theta\sigma, t^*} (x^*,x^*)$, with $t^*= c(\sigma)+2 $ of cardinality $\geq  \beta /\varepsilon^{\theta }$.
%
Let again $u$ and $v$ be traces of runs from an initial state to $(q,x^*)$ and from $(q,x^*)$ to a final state, respectively and $\nu=\length(uv)$.
In turn, $\Sep_T=u\Sep'^{\lfloor (T-\nu)/t^*\rfloor}v$ is an $\varepsilon$-separated set within $L_T(\aut)$ and provides a lower bound on the bandwidth
$$
\bandc_\varepsilon(L)\geq \limsup_{T\to \infty} \frac{\log\#\Sep_T}{T}.
$$
The right-hand side equals 
\begin{multline*}
\limsup_{T\to \infty} \frac{\left\lfloor \frac{T-\nu}{t^*}\right\rfloor \log\#\Sep'}{T}\geq
\limsup_{T\to \infty} \frac{\left(\frac{T-\nu}{t^*} -1\right)(\log\beta+\theta\log(1/\varepsilon))}{T}=\\
\frac{\log\beta+\theta\log(1/\varepsilon)}{ c(\sigma)+2}=\log(1/\varepsilon)\cdot (A_1+A),
\end{multline*}
with 
$$
A_1=\frac{\log\beta}{\log(1/\varepsilon)\cdot(c(\sigma)+2)}\geq -1
$$
for $\varepsilon$ small enough. 
From the three inequalities above, we conclude that \linebreak
 $\bandc_\varepsilon(L(\aut))\geq(A-1) \log(1/\varepsilon)$. Since it holds for any $A>0$, the language cannot be normal (nor meager), and thus, due to the classification theorem \cite[Thm.~5, Cor.~26]{3classes}, $L$ is obese.
\qed
\end{proof}

\section{Upper bound}
We recall the statement of the upper bound, proved in \cref{sec:idem,sec:fac,sec:follow,sec:0,sec:net} below.  
\propUpperBound*

The degenerate case of $\alpha=0$ is studied in  \cref{sec:deg}. 

The proof below for $\alpha >0$ is quite involved, and we provide a walkthrough to the motivated reader:
\begin{itemize}
\item In \cref{sec:idem} we perform some preparatory work. In particular, we identify conservative cyclic runs, intuitively in such runs no energy is dissipated (Lyapunov functions as in \cref{lem:puri:Lyap} remain constant). We also define  $\varepsilon$-conservative runs, admitting a small ($<\varepsilon$ over the whole lifespan) energy leak, and show that the latter can be approximated by the former in \cref{lem:eps:con}.  

\item \cref{lem:fac} is the pivot of the proof. Combining Simon's theorem with Puri's reachability, we factorize every run in long $\varepsilon$-conservative cycles with a bounded number of arbitrary intermediate fragments of a bounded total length.  
\item In \cref{lem:realize} we show that conservative runs ``follow routes'' in the refined graph $G$. Next, for each route in $G$ we describe some properties of runs following it.
\item In \cref{sec:0} we study ``pathological'' loops in time $0$ and show that they can be safely ignored (\cref{lem:fluff}). Without such loops, the number of routes to consider grows moderately (\cref{lem:count:routes}), and a bound on the number of adjacent $0/0$ edges in $G$ holds.
\item Finally, in \cref{sec:net} we use the preceding lemmas to build a series of $\varepsilon$-nets, more and more general, culminating with a required net for the whole language $L_T(\aut)$.  
\end{itemize}

\subsection{On reachability along idempotent cycles}\label{sec:idem}
Let $\sigma$ be a cycle from $q$ to itself in $\aut$ with idempotent orbit graph  $e=\gamma(\sigma)$ (it is a graph on vertices of $S(q)$). We decompose $e$ into a DAG of SCCs (which are cliques and singleton transient vertices). The SCCs can be topologically sorted, that is enumerated $C_1,\dots,C_k$ satisfying the property:
\begin{equation}\label{eq:top:sort}
 i<j \Rightarrow C_i \text{ unreachable from }C_j \text{ in } e.   
\end{equation}

We introduce a variant of barycentric coordinates (\cref{def:bary}) with respect to active faces $\conv{C_i}$.  Any point  $\x\in S(q)$ can be represented as $\x=\sum_i\lambda_i \x_i$  with  $\x_i\in \conv (C_i)$, and $ \lambda_i\geq 0$ such that $\sum_i\lambda_i=1$.

Intuitively, face-barycentric coordinates $\vec{\lambda}(\x)=(\lambda_1,\dots,\lambda_k)$ describe position of $\x$ with respect to faces spawned by each of $C_1,\dots,C_k$. We remark that weight $\lambda_i$ of a face   $\conv{C_i}$ in face-barycentric coordinates of $\x$ is merely the sum of weights of its vertices in barycentric coordinates.

Let $\succeq_e$ (in most cases we omit the subscript $e$) denote the lexicographic order over vectors  $(\lambda_1,\dots,\lambda_k)$. 

The order $\succeq$ is tightly related to reachability. 

\begin{lemma}\label{lem:lyap:faces}
For any path $\sigma\in \gamma^{-1}(e)$ and points $\x^1,\x^2\in S(q)$, if $(\x^1,\x^2)\in \overline{\mathop{Reach}_\sigma}$  then  $\vec\lambda (\x^1)\succeq \vec\lambda (\x^2)$.
\end{lemma}

\begin{proof}
The result follows from \cref{lem:puri:Lyap}. \qed
\end{proof}

\defConservativeAndEps*




The following technical result says that ``most of the weight'' of any $\varepsilon$-conservative run goes from  clique SCCs to themselves.

\begin{lemma}\label{lem:diag}
Let $(q,\x^1)\Trans\rho(q,\x^2)$ be an $\varepsilon$-conservative $e$-run decomposed  as described in \cref{lem:puri:proj}: $\vec\rho=\sum_{mn}\mu_{mn}\vec\rho_{mn}$ with $(q,v_m)\Trans{\rho_{mn}}(q,v_n)$. 
Then, whenever   $v_m\in C_i$ and $v_n\in C_j$, it holds that
$$
\left[i\neq j \lor (C_i \text{ or } C_j\text{ is a transient})\right]\Rightarrow\mu_{mn}\leq K\varepsilon.
$$
\end{lemma} 
This statement also holds for conservative runs (and $\varepsilon=0$).  
\begin{proof}
Indeed, if $i>j$, then   by \eqref{eq:top:sort} $v_n$ is unreachable from $v_m$ and $\mu_{mn}=0$. Whenever $C_i=\{v_m\}$ is a transient vertex,  $v_m$ is unreachable from $v_m$ and $\mu_{mm}=0$. 

Hence, it remains to prove the bound for $\mu_{mn}$ for the case when 
for some $j^*>i^*$, $v_m\in C_{i^*}$ and $v_n\in C_{j^*}$. 
 
 To that aim, we define, for a set of vertices $D$ and a vector $\x$, the weight  $W_D(\x)$ as the sum of barycentric coordinates of $\x$ corresponding to vertices in $S$, and remark that 
\begin{equation}\label{eq:leak1}
W_D(\x^1)=\sum_{\substack{v_m\in D\\ n\in 1..d }} \mu_{mn}\text{ and symmetrically } W_D(\x^2)=\sum_{\substack{m\in 1..d \\v_n\in D}} \mu_{mn}.
 \end{equation}
 With this notation $\lambda_i(\x)= W_{C_i}(\x)$, and by $\varepsilon$-conservativity
 \begin{equation}\label{eq:leak2}
    \varepsilon> W_{C_i}(\x^1)- W_{C_i}(\x^2).
 \end{equation}
Consider the set of vertices  $D^*=
\bigcup_{i=1}^{i^*}C_i$. Adding up  \eqref{eq:leak2} for $i=1..i^*$ we get
\begin{equation}\label{eq:leak3}
i^*\varepsilon>W_{D^*}(\x^1)- W_{D^*}(\x^2).
 \end{equation}
We remark that $\mu_{mn}=0$ for any $m\not\in D^*$ and $n\in D^*$ (nothing can enter $D^*$ from outside), due to \eqref{eq:top:sort}, and deduce from \eqref{eq:leak1}:
\begin{equation}\label{eq:leak4}
W_{D^*}(\x^1)=\sum_{\substack{v_m\in D^*\\ v_n\in D^* }} \mu_{mn}+\sum_{\substack{v_m\in D^*\\ v_n\not\in D^* }} \mu_{mn} \text{ and } W_{D^*}(\x^2)=\sum_{\substack{v_m\in D^* \\v_n\in D^*}} \mu_{mn}.
 \end{equation}
We upper-bound the left-hand side of \eqref{eq:leak3} and  rewrite the right-hand side using \eqref{eq:leak4}:
$$
K\varepsilon>\sum_{\substack{v_m\in D^*\\ v_n\not\in D^* }} \mu_{mn}\geq \mu_{i^*j^*},
$$ 
which concludes the proof. \qed
\end{proof}
Now we can prove an important fact. 
%
\epsilonCon*
\begin{proof}
Given an $\varepsilon$-conservative  run $\vec\rho=\sum_{mn}\mu_{mn}\vec\rho_{mn}$,
we build new coefficients $\mu'_{mn}$ by replacing by $0$ all those that are small by virtue of \cref{lem:diag}.  Formally:
$$
\mu'_{mn}= 
\begin{cases}
C\mu_{mn} & \text{if }  v_n \text{ and  } v_m \text{ are in the same clique;}\\
0 & \text{otherwise.} 
\end{cases}
$$
The normalizing coefficient $C$ is chosen in such a way that $\sum \mu'_{mn}=1$, it satisfies $1\leq C\leq (1-K^3\varepsilon)^{-1}\leq 1+K^3\varepsilon$, and we have that in both cases ($v_n$  and  $v_m$ 
are in the same clique or not) $|\mu'_{mn}-\mu_{mn}|\leq K^3\varepsilon$.  

The run  $\rho'=\displaystyle\sum_{m,n} \mu'_{mn}\rho_{mn}$ is conservative 
and $O(\varepsilon)$-close to $\rho$:
$$
||\vec\rho'-\vec\rho||_\infty=  \left\|\sum_{m,n} (\mu'_{mn}-{\mu_{mn}})\vec\rho_{mn}\right\|_\infty\leq K^3\varepsilon\sum_{m,n}||\vec\rho_{mn}||_\infty\leq K^5 M\varepsilon,
$$
where $M$ is the bound of all clocks in the \CTA. 

Traces  of $\rho$ and $\rho'$ are also $O(\varepsilon)$-close due to  \cref{lem:heart}. \qed
\end{proof}

\subsection{Factorization}\label{sec:fac}
We recall that by \cref{thm:simon}, each path in $\aut$ belongs to the level  $X_H$ of the sequence \eqref{eq:simon},
hence the same is true for runs of $\aut$, and the recurrence holds:
$$R_0=R_\Delta\cup\{\epsilon\};\qquad 
R_{h+1}=R_h\cdot R_h\cup\bigcup_{\substack{e\in O\\ee=e}}(R_h\cap R_e)^*
$$
with $R_\Delta$ standing for single-transition runs and $R_e$ for $\bigcup_{\pi\in \gamma^{-1}(e) }\Runs_\pi$.  Again $R_H$ contains all the runs of $\aut$.

\simonDecompConserv*



\begin{proof}
We will prove by induction, that for timed words that are traces of $R_h$ at each level $h$ there exists a required factorization with $n\leq b^h $ and $\sum\length(u_i)\leq b^h$ (with $b\geq 2$ depending only on $\varepsilon$ introduced below), and hence we can take $f(\varepsilon)=b^H$. 

For the base case $R_0$, we take the trivial factorization $w=u_0$.

For the first inductive case of $w=w_1 w_2$ with $w_i$ trace of a run 
$\rho_i\in R_h$, we obtain a factorization of $w$ by concatenating those of $w_1$ and $w_2$, and obtaining $n$ and $\sum\length(u_i)$ dominated by $2b^h\leq b^{h+1} $.

The case of $w$ trace of a run $\rho=\rho_0\dots \rho_m$ with all paths $\rho_j\in (R_h\cap R_e)$ for some idempotent $e$ is more involved. Let $q,C_i,\succ, \vec{\lambda}$ be as in \cref{sec:idem}.
The run $\rho$ has the form
$$
(q,\x_0)\Trans{\rho_0}(q,\x_1)\Trans{\rho_1}\cdots \Trans{\rho_m}(q,\x_m),
$$
and by \cref{lem:lyap:faces},
$$\vec{\lambda}(\x_0)\succeq\vec{\lambda}(\x_1)\succeq\cdots\vec{\lambda}(\x_m),$$
and thus (for some potentially very big $m$)
$$\lfloor\vec{\lambda}(\x_0)/\varepsilon\rfloor\succeq\lfloor\vec{\lambda}(\x_1)/\varepsilon\rfloor\succeq\cdots\lfloor\vec{\lambda}(\x_m)/\varepsilon\rfloor.$$
The latter is a non-increasing (in some total order) sequence within the finite set $\{0..\lfloor1/\varepsilon\rfloor\}^k$ (with cardinality bounded by $b=(1+1/\varepsilon)^K$). Such a sequence always can be partitioned into  $N\leq b$ chunks (corresponding to distinct values of $\lfloor\vec{\lambda}(\x_i)/\varepsilon\rfloor$); within each chunk, this value does not change, but it goes down w.r.t.~$\succ$ from every chunk to the next one.   Grouping $\rho_i$ belonging to the same chunk, we obtain a new, rougher factorization of the run $\rho$.  
$$
(q,\y_0)\Trans{\rho'_0}(q,\y'_0)\Trans{\rho''_0}(q,\y_1)\Trans{\rho'_1}(q,\y'_1)\cdots(q,\y_j)\Trans{\rho'_j}(q,\y'_j),
$$
such that $\y_0=\x_0$ and $\y'_j=\x_j$, satisfying the properties:
\begin{itemize}
    \item its size can be bounded: $j\leq b$; 
    \item $\lfloor\vec{\lambda}(\y_i)/\varepsilon\rfloor=\lfloor\vec{\lambda}(\y'_i)/\varepsilon\rfloor$ and thus $\rho'_i$ is $\varepsilon$-conservative for all $i$;
    \item each of $\rho''_i$ belongs to  $R_h$, and by inductive hypothesis its trace can be factorized into $\varepsilon$-conservative words and at most $b^h$ others with cumulated duration at most $b^h$. 
\end{itemize}
This factorization of each $\rho''_i$ yields a required factorization of $w$ into $\varepsilon$-conservative words and at most $b^{h+1}$ others with cumulated duration at most $b^{h+1}$. This concludes the inductive proof.\qed
\end{proof}
\subsection{On runs following routes}\label{sec:follow}
In this section, we relate some runs of the automaton with routes in $G$, explore such runs, and show that all conservative runs belong to this category. 
\defRoute*
\runFollowsRoute*


We remark that the previous definition is quite restrictive, however it captures all conservative runs:
\conservRunFollowsRoute*
\begin{proof}
Consider a conservative run $\rho$ from $(q,\x)$ to $(q,\y)$ along a path $\sigma$ in $\aut$ with $\gamma(\sigma) =e$ and $e^2=e$. Let again the vertex sets $C_1,\dots,C_k$ be all SCCs of $e$ as in the previous section. By definition of conservative run, there exist $\x_i, \y_i$ such that $\x_i,\y_i\in \conv{C_i}$ and $\x=\sum_{i=1}^k \lambda_i \x_i$ and $\y=\sum_{i=1}^k \lambda_i \y_i$. 
We define the set of active faces as
 $\AF=\{i| \lambda_i>0\}$.  By virtue of \cref{lem:diag}, $\AF$ contains only indices of some cliques, not the transient vertices.

Let us decompose the run as described in \cref{lem:puri:proj}: $\vec\rho=\sum_{mn}\mu_{mn}\vec\rho_{mn}$ with $(q,v_m)\Trans{\rho_{mn}}(q,v_n)$. Using \cref{lem:diag} (with $\varepsilon=0$)  we conclude that  $\mu_{mn}>0$ only when $v_m$ and $v_n$ belong to some clique $C_i$ (the same).

We have that 
$$
\lambda_i=W_{C_i}(\x)=\sum_{v_m,v_n\in C_i} \mu_{mn}. 
$$
Defining (for $i\in \AF$) 
$$
\vec\rho_i=\frac1{\lambda_i}\sum_{v_m,v_n\in C_i} \mu_{mn}\vec\rho_{mn}
$$
we obtain a useful decomposition 
\begin{equation}\label{eq:proj}
    \vec\rho=\sum_{i\in \AF} \lambda_i\vec\rho_i. 
\end{equation}

For any $j<|\sigma|$ let us factorize $\sigma$ into a prefix $\sigma_1$ of length $j$ from $q$ to $p^j$ and a suffix  $\sigma_2$ from $p^j$ to $q$. Let $g_1=\gamma(\sigma_1)$ and $g_2=\gamma(\sigma_2)$ their orbit graphs.

For each $i\in \AF$ let $C_i^j$ be the set of vertices of $v\in S(q^j)$ such that $g_1$ contains an edge from $C_i$ to $v$   and $g_2$ contains an edge from $v$  to $C_i$.

Let us prove that each $C^j_i$ is recurrent. Indeed, let $v,w\in C^j_i$, then by construction there exist $v_0,w_0\in C_i$ and edges  $v_0 \to v$ in $\gamma(\sigma_1)$  and  $w\to w_0$ in  $\gamma(\sigma_1)$. Since $C_i$ is a clique of $\gamma(\sigma)$, this orbit graph contains a path from $w_0$ to $v_0$. We can now conclude that there exists an edge from $v$ to $w$ in $g=\gamma(\sigma_2\sigma\sigma_1)$, and $C^j_i$ is a clique of the orbit graph $g$.

Let us prove that $C^j_i$ are disjoint. Indeed, suppose that some vertex $v$ belongs to both  $C^j_i$ and $C^j_k$. By construction, there exists a run (along $\sigma_1$) from a vertex of $C_i$ to $v$ and another from $v$ to a vertex of $C_k$, which implies a run along $\sigma$ from a vertex of $C_i$ to one of $C_k$. For the same reason, there exists a run from $C_k$ to $C_i$.  Since those are maximal SCCs, we conclude that $k=i$.

We conclude that $(q_j,(C^j_i, i\in \AF))$  is a legal vertex of the refined graph $G$. 

We observe now that each projection run $\rho_i, i\in \AF$ passes through
 all the $(q_j,\conv(C^j_i))$. Thus $\rho$ follows the route $\sigma,\sigma'$  with $\sigma'=(q_j,(C^j_i, i\in A))_{j\in 1..n}$ as required. \qed
\end{proof}

Now, let us fix a cyclic route $(\sigma,\sigma')$ and consider a  run $\rho$ following it and its trace $w=\Word(\rho,t_0)$, with standard notations for their components:
\begin{equation}\label{eq:all}
\begin{split}
 \sigma&=q_0\trans[a_1]{\delta_1}q_1\trans[a_2]{\delta_2}\cdots q_n\text{ with }q_0=q_n;\\
 \sigma'&=(q_0,\ell_0)\trans{r_1/c_1}(q_2,\ell_2)\trans{r_2/c_2}\cdots (q_n,\ell_n);\\
 \rho&=(q_0,\x_0)\trans[a_1,d_1]{\delta_1}(q_1,\x_1)\trans[a_2,d_2]{\delta_2}\cdots (q_n,\x_n); \\
 w&=(a_1,t_1)(a_2,t_2)\cdots(a_n,t_n)\text{ with }t_0\in [\tau,\tau+1), \tau\in \real_+;\\ &\hspace{15em}\text{ and } t_i=t_{i-1}+d_i \text{ for }i\in 1..n.
\end{split}    
\end{equation}

The following technical lemma describes the quantities in  $w$ and $\rho$ and will be instrumental  in the construction of the $\varepsilon$-net.

\begin{lemma}\label{lem:quantity}
In conditions of \eqref{eq:all} the following holds:
\begin{enumerate}
   \item if $t_i-t_0>1$ and  $a_i= b$ (the heartbeat), mthen $t_i= t_{\prevb_i}+1$; 
    \item  $t_i\in [\tau+\countb_i-1,\tau+\countb_i+2]$;
    \item if $t_i-t_0>M$, then $\x_i^k=t_i-t_{\lastreset_i^k}$;
    \item whenever  $r(\delta')=0$, the duration of this transition $d_i=g_i(\x_{i-1})$.
\end{enumerate}

Here $\prevb_i$ is the index\footnote{previous beat} of the last hearbeat event within $a_1,\dots,a_{i-1}$; also   $\countb_i$ stands for the number\footnote{count of beats} of heartbeat events $b$ within $a_1,\dots,a_i$, and $\lastreset_i^k$ for the maximal $j\leq i$ such that $\delta_j$ resets $x^k$ (it exists since all clocks are bounded by $M$). Finally, $g_i$ are polynomial functions (described below in the proof). Functions $g_i$ are uniformly Lipshitz continuous: there exist $\nu>0$ depending only on $\aut$ such that
$|g_i(\x)-g_i(\y)|\leq \nu||\x-\y||_\infty$. 
\end{lemma}
\begin{proof}
(1) and (2) follow from the definition of the heartbeat. 
(3) is always true in timed automata whenever ${\lastreset_i^k}$ is well-defined. This is always the case for $t_i>M$ since all clocks are bounded by $M$.

Let us prove (4). As in \cref{def:follow}, let $\vec\rho=\sum\lambda_j\vec\rho_j$
 and let $\x_{i-1}$ decompose as $\sum_j\lambda_j{\z_j}$. Since the reward is $0$, there exist only one possible timing for a transition $\delta$ from $(q_{i-1},\z_j)$ to $(q_i,\ell_{ij})$, we denote it $t_j$, and $\rho_j$ must respect this timing. Thus necessarily $t=\sum_j\lambda_j t_j$, and $g_i$ will map $\x$ to this $t$.
 
 To understand the exact analytic form of $g(\x)$, we observe that $\lambda_j$ are affine combinations of $x_j$ and that $t_j$ has one of the forms: $c,x_k-c,c-x_k$. Thus $g(\x)$ is a polynomial of degree $\leq 2$ with uniformly bounded coefficients. 
\qed 
\end{proof}
\subsection{On zero loops and proper routes}\label{sec:0}
We consider now loops feasible in $0$ time. As usual in the theory of timed automata, such Zeno cycles  require special attention. 
For a  one-step route $\delta,\delta'$ of cost and reward $0$ with $(p^1,\ell^1)\trans{\delta'}(p^2,\ell^2)$ we define its 0-orbit graph $\gamma_0(\delta,\delta')=\langle p^1,\ell^1,G,p^2,\ell^2\rangle$ with $G$ containing only those edges $(v^1,v^2)$ of $\gamma(\delta)$ that are compatible with $\delta'$, i.e.  $v^1\in\ell^1_i$ and $v^2\in\ell^2_i$ with the same $i$. 

We put $\gamma_0(\delta,\delta')=0$ for transitions with positive cost or reward, and we propagate $\gamma_0$ to all nonempty routes of cost and reward $0$ as a semigroup morphism, similarly to \cref{def:orbit:long}. For any $(q,\ell)\in G$ we define $1_\ell=\langle q,\ell,\mathbf{Id}_{\cup\ell},q,\ell\rangle$ with $\mathbf{Id}$ the graph consisting of self-loops on all vertices in $\cup\ell$. 

We denote the  image of $\gamma_0$ augmented with $0$  by $O_0$, it is a semigroup.   
Let us state some  useful properties of $O_0$.
\begin{lemma}\label{lem:0:monoid}
The finite semigroup $O_0$ has the following properties:
\begin{enumerate}
    \item for each non-zero $a\in O_0$ (image of a route from $(p^1,\ell^1)\to(p^2,\ell^2)$ with cost and reward 0), its graph corresponds to a surjective function  $\ell^1_i \to \ell^2_i$ for every $i$;
    \item for image of a cycle from $(p,\ell)\to(p,\ell)$ (with cost and reward 0), its 0-orbit graph corresponds to a permutation  $\ell_i \to \ell_i$ for every $i$;
    \item if $ab=a\neq 0$ for $a,b\in O_0$, then $b=1_\ell$ for a certain  $\ell$. 
\end{enumerate}
\end{lemma}
\begin{proof}
\begin{enumerate}
    \item Any vertex of active face  $\ell^1_i$ has a  successor vertex in $\ell^2_i$ by \cref{con:main}; this successor is unique  (otherwise the reward would be $>0$). As for surjectivity,  every   vertex in $\ell^2_i$ has a predecessor also by \cref{con:main}. 
    \item A surjective function from a finite set to itself is always a permutation. 
    \item Orbit graphs $a$ and $b$ correspond to some routes from $(p^1,\ell^1)$ to $(p^2,\ell^2)$ and
    from $(p^2,\ell^2)$ to $(p^2,\ell^2)$, respectively. Let $v$ be a vertex in $\ell^2$, vertex  $v'$ its predecessor via $a$ in $\ell^1$, and $v''$ its successor via $b$ in $\ell^2$.  Since $ab=a$, we have that $v''=v$. The only permutation graph satisfying this property for all $v$ is $1_{\ell^2}$. We conclude that $b=1_{\ell^2}$. \qed    
\end{enumerate}
\end{proof}

\begin{definition}
A \emph{$0$-loop} is a cyclic route $\sigma,\sigma'$ from $(q,\ell)$ to itself with $c(\sigma')=r(\sigma')=0$, such that $\gamma_0(\pi)=1_\ell$.
\end{definition}

\begin{lemma}
Any run following a $0$-loop has duration $0$ and cycles from a state $(q,\x)$ to exactly the same state.
\end{lemma}
\begin{proof}
We remark that the heartbeat clock $h$ is never reset in such a run, and this clock is one of the coordinates of each vertex  $v$ of active faces of $S(q)$. The only run respecting $\gamma_0(\pi)=1_\ell$ from $v$ leads to $v$ and takes $0$ time (otherwise the value of   $h$ would change). The run of $(q,\x)$ following the route, by \cref{def:follow} is a convex combination of runs of active vertices, and as such goes exactly to $(q,\x)$ in $0$ time. \qed
\end{proof}

Routes of cost and reward $0$ satisfy the following pumping lemma. 
\constBforZeroLoop*
\begin{proof}
Let $B=\# O_0$.  For the route $(\pi,\pi')$ denote its factor  between positions  $i$ and $j$ by $(\pi,\pi')_{i..j}$ and consider its $\gamma_0$-image $e_{i..j}=\gamma_0((\pi,\pi')_{i..j}$. By pigeonhole principle, for some $i<j$ we have $e_{0..j}=e_{0..i}$, hence  $e_{0..i}e_{i..j}=e_{0..i}$, and by \cref{lem:0:monoid} $e_{i..j}=1_\ell$, and we have found a 0-loop $(\pi,\pi')_{i..j}$. \qed
\end{proof}

Routes without $0$-loops have some interesting properties and will be instrumental in the construction of $\varepsilon$-net in the next section. 

\defProperRoute*



\begin{corollary}\label{cor:cost}
In proper routes of normal automata, all factors of cost and reward $0$ have length $\leq B$. 
\end{corollary}

The last corollary  is immediate from \cref{lem:pumping} and helps count proper routes:
\countNumberCyclicRoutes*

\begin{proof}
Let $\alpha$ be the maximal reward-to-cost ratio in $\aut$, by \cref{lem:obese} it cannot be infinite. 

Let $\sigma,\sigma'$ be such a proper route from $q$ to itself. By \cref{cor:cost}, at least every $B$ steps, $\sigma'$ has  a transition of cost $1$ or reward $1$.  Thus $$
|\sigma|=|\sigma'|\leq (c(\sigma')+r(\sigma')+1)B\leq ((\alpha+1)T+1)B=K_{10}T+K_{11}.
$$ 
The cardinality of all such routes cannot exceed $\left(N\cdot 2^{\#\Sigma}\#G\right)^{K_{10}T+K_{11}}$, which admits a required upper bound. \qed
\end{proof}

We show how to approximate traces of runs following arbitrary routes by those following proper ones.
\begin{con}
Let $w$ be a timed word over $\Sigma$, its \emph{fluffing} $\phi(w)$ is the set of all words obtained from $w$ by inserting after each letter $(a,t)$ any sequence $(c_1,t)\dots(c_m,t)$ with $c_i\in\Sigma$ distinct elements in alphabetic order. 
\end{con}
We remark that $\#\phi(w)\leq 2^{|w|\#\Sigma}$.  

\fluffFollowingRoute*
%
\begin{proof}
We need a notation: given a run $\rho$, let $\mathop{letters}(\rho)$ be the timed word of duration $0$ containing all the events of $\Sigma$ occurring in $\rho$ and sorted (without duplicates) in alphabetic order. 

Let $\rho$ be  a run, it can be written in the form $\xi_1\zeta_1\dots\xi_n\zeta_n\xi_{n+1}$, where the $\zeta_i$ follow $0$-loops and the $\xi_i$ don't contain factors along $0$-loops. Then $\hat\rho\triangleq\xi_1\dots\xi_n\xi_{n+1}$ is a run (because the $\zeta_i$ are cycles), which is proper, by definition. The labelings of $\hat\rho$ and $\rho$ are respectively $\hat w=\Word(\xi_1)\dots\Word(\xi_n)\Word(\xi_{n+1})$ and $w=\Word(\xi_1)\Word(\zeta_1)\dots\Word(\xi_n)\Word(\zeta_n)\Word(\xi_{n+1})$.
Now let us define $\tilde w\triangleq\Word(\xi_1)\mathop{letters}(\zeta_1)\dots\Word(\xi_n)\mathop{letters}(\zeta_n)\Word(\xi_{n+1})\in \phi(\hat w)$.  Since $\Word(\zeta_i)$ has duration $0$, $d(w,\tilde w)=0$ and thus $d(w, \phi(\hat w))=0$.
%
%
%
\qed
\end{proof}

\subsection{Building $\varepsilon$-nets}\label{sec:net}

\epsNetWordsOneRoute*
Note that the above constraints are relaxations of the properties in Lemma \cref{lem:quantity}.


\lemmaEpsNetWordsOneRoute*

\begin{proof}
A word $w =(a_1,t_1)\dots(a_n,t_n)$ following the route can be $\varepsilon$-approximated by $v=(a_1,s_1)\dots(a_n,s_n)$ with $s_i=\varepsilon\lfloor t_i  /\varepsilon\rfloor$. By \cref{lem:quantity} applied to $w$, we deduce that $v$ satisfies the constraints of \cref{con:con:1} and thus $v\in \Net^{\mathbf{fol}}_\varepsilon(\pi,\pi',\tau)$. \qed
\end{proof}



\epsFollowProperRoutes* 


\indeedEpsFollowProperRoutes*

\begin{proof}
The statement follows from \cref{lem:net:con:1}. \qed
\end{proof}

\epsNetConservWords*


\coroEpsNetCoservWords* 

\begin{proof}
The first statement follows from \cref{lem:net:con,lem:fluff}  and the fact that $dur(w) < \Theta$ implies that the cost of its run $< \Theta + 1$. 
The statement about $\varepsilon$-conservative words is now immediate by \cref{lem:eps:con}. \qed
\end{proof}




\epsNetAut* 

%



\lemmaIndeedEpsNetAut* 

\begin{proof}
Immediate from the construction, \cref{lem:fac,lem:net:con}. \qed
\end{proof}

\propEpsCount*
\begin{proof}[of \eqref{eq:count:fol}]  All the elements of the net have the same untiming $a_1\dots a_n$, so we are interested in the number of timings.
We remark that in the route/the word:
\begin{itemize}
    \item there are $r(\pi')$ edges of reward 1;
    \item there are $c(\pi')$ heartbeat edges;
    \item there are $\leq r(\pi')+ c(\pi')$ contiguous blocks of edges of reward $0$ without heartbeat transitions (after an edge of reward $1$ or after a heartbeat transition);
    \item each such block contains at most $B$ edges (see \cref{cor:cost});
    \item by \cref{lem:count:routes} there are at most $K_{10}M+K_{11}$ events during first $M$ time units.
\end{itemize}
Thus in $\Net^{\mathbf{fol}}_\varepsilon(\pi,\pi',\tau)$ there are 
\begin{itemize}
    \item $3/\varepsilon$ possible timings for each edge within first $M$ time units and each edge of reward $1$;
    \item (after $M$ time) only three possible timings for each heartbeat edge;
    \item (after $M$ time) in each contiguous block of edges of reward $0$ without heartbeat transitions, there are $2\nu+3$ possible time values for the first edge, followed by $2\nu+3$ possible time values for the second one, etc. This yields at most $(2\nu+3)^B$ possible timings for the block. 
\end{itemize}
This leads to the following estimate for the net cardinality:
$$\#\Net^{\mathbf{fol}}_\varepsilon(\pi,\pi',\tau)\leq(3/\varepsilon)^{r(\pi')+K_{10}M+K_{11}}\cdot 3^{c(\pi')} \cdot  (2\nu+3)^{B(r(\pi')+c(\pi'))}.
$$
To simplify it, we pass to logarithms, introduce new constants and use the fact that the reward-to-cost ratio in any cycle $\leq \alpha$ and thus $r(\pi')\leq \alpha c(\pi')$. 
$$
\log\#\Net^{\mathbf{fol}}_\varepsilon(\pi,\pi',\tau)\leq
(\alpha c(\pi')+K_3)\log\frac1\varepsilon+ K_5c(\pi')+K_{12}. \qquad\qed
$$
\end{proof}
\begin{proof}[of \eqref{eq:count:cons}]
It is immediate from \cref{eq:count:fol,lem:count:routes} that
$$
\log\#\Net^{\mathbf{fol}}_{\varepsilon,\Theta}(\tau)\leq (\log K_1 +K_2\Theta)+ (\alpha \Theta+K_3)\log\frac1\varepsilon+ K_5 \Theta+K_{12}. 
$$
Fluffing  operation $\phi$ maps every word of $n$ letters into  $<2^{n\#\Sigma}$ words, in our case $n\leq K_{10}(\Theta+1)+K_{11}$ (see \cref{lem:count:routes}), thus
\begin{multline*}
\log\#\phi\left(\Net^{\mathbf{fol}}_{\varepsilon,\Theta+1}(\tau)\right)\leq \log\#\Net^{\mathbf{fol}}_{\varepsilon,\Theta+1}+(K_{10}(\Theta+1)+K_{11})\#\Sigma= \\ \alpha \Theta \log\frac1\varepsilon + K_3\log\frac1\varepsilon + K_6 \Theta +K_7.\qquad\qed
\end{multline*}
\end{proof}
\begin{proof}[of \eqref{eq:count:all}]
The number of breakpoint sequences $\zeta$ is bounded by
$(T/\varepsilon)^{f_\aut(\varepsilon)}$.

For a given breakpoint sequence $\zeta$,
the $\log$ of cardinality of a universal $\varepsilon$-net on $[a_i,b_i]$ is bounded by ${K_8(b_i-a_i)/\varepsilon}$, see for example \cite{distance}.
The $\log$ of cardinality of the conservative net on $[b_i,a_{i+1}]$ is given by   \eqref{eq:count:cons} with
$\Theta=a_{i+1}-b_i$.
Also, $\sum_{i=0}^{n}(b_i-a_i)\leq f_\aut(\varepsilon)$ by \cref{lem:fac}, and $\sum_{i=0}^{n}(a_{i+1}-b_i)\leq T$.
Thus the cardinality of $\Net_\varepsilon^\zeta$ can be bounded  as follows:
\begin{multline*}
\log\#\Net_\varepsilon^\zeta\leq
\sum_{i=0}^{n}\frac{K_8(b_i-a_i)}{\varepsilon} +\\
\sum_{i=0}^{n-1}\left(\alpha (a_{i+1}-b_i) \log\frac{K_0+1}{\varepsilon} + K_3\log\frac{K_0+1}{\varepsilon} + K_6 (a_{i+1}-b_i) +K_7\right)\leq\\
\frac{K_8f_\aut(\varepsilon)}{\varepsilon}+
\alpha T \log\frac1\varepsilon+ \alpha T \log (K_0+1) + K_3f_\aut(\varepsilon)\log\frac{K_0+1}\varepsilon + K_6 T +K_7f_\aut(\varepsilon)=\\
\alpha T \log\frac1\varepsilon + K_9 T+ h_\aut(\varepsilon)
\end{multline*}
for some function $h_\aut$. 
Multiplying by the number of possible $\zeta$ sequences, we get
\[
\log\#\Net_{\varepsilon,T}\leq
f_\aut(\varepsilon)\log(T/\varepsilon)+\alpha T \log\frac1\varepsilon + K_9 T+ h_\aut(\varepsilon)
\]
and thus
$$
\limsup_{T\to\infty}
 \frac{\log\#\Net_{\varepsilon,T}}{T} \leq 
\lim_{T\to\infty} \frac{ f_\aut(\varepsilon)\log\frac T\varepsilon+\alpha T \log\frac1\varepsilon + K_9 T+ h_\aut(\varepsilon)}T
=
\alpha \log\frac1\varepsilon+K_9.
$$
\qed
\end{proof}

\Cref{lem:upper} is now immediate from \cref{lem:it:is:a:net,eq:count:all}. 

\subsection{The extreme case}\label{sec:deg}
We remark that the constructions and lemmas of \cref{sec:net} still hold in the case of  $\alpha=0$.
\propAlphaZero*
\begin{proof}
If the graph $G$ is acyclic, then by \cref{lem:realize} there are no conservative cycles, and by \cref{lem:fac} the duration of all words in $L(\aut)$ is bounded by $f(\varepsilon)$ for $\varepsilon$ small enough. Hence  $L(\aut)$ is meager.

Otherwise, the result follows from \cref{lem:it:is:a:net,eq:count:all}. \qed
\end{proof}